\documentclass[11pt, a4paper]{article}

\renewcommand{\mid}{\hspace{.5mm}\vert\hspace{.5mm}}

\usepackage{helvet}         
\usepackage{courier}        
\usepackage{type1cm}        
\usepackage{multicol}        
\usepackage[bottom]{footmisc}
 
\usepackage{geometry} 

\usepackage{dsfont} 
\usepackage[dvipsnames]{xcolor}

\usepackage{graphicx} 
\graphicspath{{figures/}}
\usepackage[font=scriptsize]{caption}

\usepackage{booktabs} 

\usepackage{array} 

\usepackage{paralist} 

\usepackage{verbatim} 

\usepackage{subfig} 

\usepackage{amsmath} 
\usepackage{mathrsfs,amssymb,amsfonts,amsthm,natbib,amsbsy}
\usepackage{thmtools, thm-restate}
\usepackage{amssymb}
\usepackage{bm}
\usepackage{mathtools}
\usepackage{econometrics}

\usepackage[colorlinks=true, urlcolor=blue, linkcolor=blue, citecolor=blue, pdfborder={0 0 0}]{hyperref}
\usepackage{cleveref}
\usepackage{autonum} 
\usepackage[normalem]{ulem}

\usepackage[acronym]{glossaries}
\setacronymstyle{long-short}
\makeglossaries

\newacronym{CIR}{CIR}{Cox--Ingersoll--Ross}
\newacronym{WF}{WF}{Wright--Fisher}
\newacronym{FFBS}{FFBS}{Forward Filtering Backward Sampling}
\newacronym{MCMC}{MCMC}{Monte Carlo Markov Chain}
\newacronym{HMM}{HMM}{Hidden Markov Model}
\newacronym{MH}{Metropolis--Hastings}{Metropolis--Hastings}

\usepackage{sidecap} 

\usepackage{graphicx} 

\usepackage{multicol} 

\usepackage[bottom]{footmisc}
\usepackage[ruled,vlined]{algorithm2e}

\newtheorem{assumption}{Assumption}

\newtheorem{theorem}{Theorem}
\newtheorem{proposition}[theorem]{Proposition}

\def\d{\mathrm{d}}

\def\P{\mathbb{P}}

\def\X{\mathcal{X}}

\def \Y{\mathcal{Y}} 

\def\M{\mathcal{M}} 

\def\Th{\varTheta} 

\def\Z{\mathbb{Z}_+}

\def\r{r} 

\def\N{\mathbb{N}}

\def \R{\mathbb{R}}

\def \E{\mathbb{E}}

\def\F{\mathcal{F}}

\def \aa {\bs\alpha}

\def \yy {\mathbf{y}}

\def \ee {\mathbf{e}}

\def \mm {\mathbf{m}}

\def \nn {\mathbf{n}}

\def \ii {\mathbf{i}}

\def \oo {\mathbf{0}}

\def \ll {\mathbf{l}}

\def \kk {\mathbf{k}}

\def \xx {\mathbf{x}}

\def \Ga {\text{Ga}}
\def \NBin {\text{NB}}
\def \Dir {\text{Dir}}
\def \DirMN {\text{DM}}

\def\d{\mathrm{d}}
\def\X{\mathcal{X}}
\def \Y{\mathcal{Y}} 
\def\M{\mathcal{M}} 
\def\Th{\varTheta} 
\def\Z{\mathbb{Z}}
\def\r{r} 

\def\N{\mathbb{N}}
\def \P{\mathbb{P}}
\def \R{\mathbb{R}}
\def \E{\mathbb{E}}
\def\F{\mathcal{F}}

\newcommand{\bs}[1]{\boldsymbol{#1}}

\newcommand{\norm}[1]{|{#1}|}
\newcommand{\B}{\mathcal{B}} 
\newcommand{\MM}{\mathrm{\bf M}} 
\newcommand{\aMM}{\overleftarrow\MM} 
\newcommand{\avt}{\!\overleftarrow\vartheta\!}

\def \aa {\bs\alpha}
\def \yy {\mathbf{y}}

\def \ee {\mathbf{e}}
\def \ii {\mathbf{i}}
\def \jj {\mathbf{j}}

\def \mm {\mathbf{m}}
\def \nn {\mathbf{n}}

\def \ii {\mathbf{i}}
\def \oo {\mathbf{0}}

\def \ll {\mathbf{l}}
\def \kk {\mathbf{k}}
\def \xx {\mathbf{x}}
\def \yy {\mathbf{y}}

\usepackage{xr}

\makeatletter
\newcommand*{\addFileDependency}[1]{
  \typeout{(#1)}
  \@addtofilelist{#1}
  \IfFileExists{#1}{}{\typeout{No file #1.}}
}
\makeatother
 
\newcommand*{\myexternaldocument}[1]{%
    \externaldocument{#1}%
    \addFileDependency{#1.tex}%
    \addFileDependency{#1.aux}%
}

\myexternaldocument{supplement}

\begin{document}

\title{\vspace{-15mm}\bf   \LARGE  Exact inference for a class of non-linear \\ hidden Markov models on general state spaces}

\author{{\sc Guillaume  Kon Kam King} \\ 
\emph{Universit\'e Paris-Saclay, INRAE, MaIAGE}\\ 
78350, Jouy-en-Josas, France \\ guillaume.kon-kam-king@inrae.fr\\
\vspace{0.5em}\\
{\sc Omiros Papaspiliopoulos}\\ 
\emph{ICREA and Universitat Pompeu Fabra} \\ Ram\'on Trias Fargas 25-27, 08005, Barcelona, Spain \\ omiros.papaspiliopoulos@upf.edu\\
\vspace{0.5em}\\
{\sc Matteo Ruggiero}\\ \emph{University of Torino and Collegio Carlo Alberto} \\
Corso Unione Sovietica 218/bis, 10134, Torino, Italy \\ matteo.ruggiero@unito.it}

\maketitle

\abstract{
\begin{quote}
Exact inference for hidden Markov models requires the evaluation of all distributions of interest - filtering, prediction, smoothing and likelihood -  with a finite computational effort. This article provides sufficient conditions for exact inference for a class of hidden Markov models on general state spaces given a set of discretely collected indirect observations linked non linearly to the signal, and a set of practical algorithms for inference. The conditions we obtain are concerned with the existence of a certain type of dual process, which is an auxiliary process embedded in the time reversal of the signal, that in turn allows to represent the distributions and functions of interest as finite mixtures of elementary densities or products thereof. We  describe explicitly how to update recursively the parameters involved, yielding qualitatively similar results to those obtained with Baum--Welch filters on finite state spaces. We then provide practical algorithms for implementing the recursions, as well as approximations thereof via an informed pruning of the mixtures, and we show superior performance to particle filters both in accuracy and computational efficiency. 
The code for optimal filtering, smoothing and parameter inference is made available in the Julia package DualOptimalFiltering.
\end{quote}
}

\newpage
\tableofcontents 

\normalsize

\section{Introduction}\label{sec:intro}

Inference for \glspl{HMM}, is a fundamental statistical problem concerned with learning the parameters and the trajectory of an unobserved Markov process, called \emph{signal}, given noisy or indirect observations, typically collected at discrete times. 
Sometimes also referred to as state-space models, \glspl{HMM} have found widespread application in a variety of frameworks that include genomics \citep{brown1993using,yau2011bayesian,guha2008bayesian,titsias2016statistical}, proteomics \citep{bae2005prediction}, time series analysis \citep{Sarkar2019}, temporal clustering \citep{crane2017hidden}, signal processing \citep{fox2011sticky},
econometrics \citep{hamilton1990analysis,chib1996calculating}, brain imaging, target tracking and animal movement \citep{quick2017hidden,langrock2015nonparametric}, to mention a few examples. General treatments of \glspl{HMM} can be found for example in \cite{chopinIntroductionSequentialMonte2020, book:54642,sarkka2013bayesian}.

Let $\{X_{t},t\ge0\}$ be a Markov process on $\X\subseteq\R^{K}$, with initial distribution $\nu_{0}$ and transition density $P^\psi_{t}(x'\mid x)$, parametrised by a finite-dimensional vector $\psi$. 
The process $\{X_{t},t\ge0\}$ is assumed to be unobserved, and referred to as \emph{hidden signal}, and to evolve in continuous time. 
Observations $Y_{t}\in \Y\subseteq \R^{D}$  are taken to be conditionally independent of everything else given the current value of the signal, to which the relate through the \emph{emission density} $f^\psi_{x_{t}}(\cdot)$, parametrised by both $\psi$ and $x_{t}$, whereby $ Y_t \mid (X_t=x_t) \overset{\text{ind}}{\sim} f^\psi_{x_{t}}(\cdot)$. We assume observations are collected at discrete times $0\le t_{0}< t_{1}<\ldots$.

In this framework, the quantities of statistical interest are typically given by the density of the signal given past observations
$p^\psi(x_{t_{i+k}}\mid y_{t_{0}},\ldots,y_{t_{i}})$, i.e., the \emph{predictive distribution}; 
or given observations up to present time $p^\psi(x_{t_{i}}\mid y_{t_{0}},\ldots,y_{t_{i}})$, 
called \emph{filtering distribution}; 
or given past, present and future observations  $p^\psi(x_{t_{i-k}}\mid y_{t_{0}},\ldots,y_{t_{i}})$, called  \emph{smoothing distribution}. 
In addition, the likelihood of the observations $p^\psi(y_{t_{0}},\ldots,y_{t_{i}})$ is a primary object of interest as well, as it allows performing inference on the model parameters $\psi$, often key quantities in an applied context. 
The likelihood is obtained here by integrating out the hidden trajectory of the signal from the joint density $p^\psi(x_{t_{0}},\ldots,x_{t_{i}},y_{t_{0}},\ldots,y_{t_{i}})$. 
Whenever it causes no confusion, we will drop the superscript $\psi$ for notational simplicity, and return to the problem of drawing inference on $\psi$ in \Cref{subsec:inference_param}.

Performing exact inference for \glspl{HMM} entails being able to compute the above distributions with a finite computational effort. Two important classes of models have long been known to allow such computation. 
The first requires the signal state space $\X$ to be a finite set, whereby the quantities of interest are obtained through the Baum--Welch filter (see \citealp{book:54642}) and elaborations thereof. 
The second is given by linear Gaussian systems (e.g.~Ornstein--Uhlenbeck signals and linearly linked Gaussian emissions), in which case all quantities of interest are obtained by updating parameters of Gaussian distributions through the celebrated Kalman--Bucy filter and elaborations thereof. 
The common feature of the two above cases is the existence of a finite-dimensional process which completely characterises the distributions of interest, so that these can be obtained by appropriately updating this process, called \emph{finite-dimensional filter}. The \emph{computational complexity} of the filter is given by the number of operations needed for performing such updates, which for finite-dimensional filters grows linearly in the number of observations. Outside the above mentioned classes, finite-dimensional filters are typically rare and difficult to obtain. 
See \cite{Ferrante1992,Ferrante1990,ferrante1998finite,Gunther1981Finite,runggaldier2001sufficient} for general conditions. 

A major breakthrough in the study of \glspl{HMM} was achieved in \cite{chaleyat2006computable} (see also \citealp{genon2003non,genon2004random}), who introduced the notion of  \emph{computable inference} for \glspl{HMM}. This applies to models for which the distributions of interest can be characterised by a finite-dimensional process whose size, however, can increase as the number of observations increases. 
In such cases, the recursive updates can be shown to have polynomial computational complexity in the number of observations. 
In this framework,  \cite{Papaspiliopoulos2014a} identified a structural property of the transition density of the signal that permits computable filtering. 
They showed that if the signal has a dual process, i.e.~an auxiliary process embedded in its time reversal, given by a pure-death process on a multidimensional grid subordinated to an ODE, then all filtering distributions are finite mixtures of parametric densities which are conjugate to the emission density of the observations. \cite{Papaspiliopoulos2016} in turn extended these results to signals given by two measure-valued processes, namely the Fleming--Viot diffusion and a branching measure-valued diffusion, and more recently \cite{ascolaniPredictiveInferenceFleming2020} applied their results to Bayesian predictive inference in a nonparametric framework. 

While \cite{Papaspiliopoulos2014a,Papaspiliopoulos2016} focused on computable filtering, the present article obtains theory and practical algorithms for the whole agenda of inference with \glspl{HMM}. 
Under a set of sufficient conditions essentially analogous to those in \cite{Papaspiliopoulos2014a}, we show that  all distributions of interest for the signal can be expressed as finite mixtures of elementary densities, and the likelihood of the observations takes the form of a finite product of mixtures. We provide explicit  recursive formulae that describe the parameters updates, and also show how to obtain samples from the joint smoothing distribution of the signal. 
Moreover, we detail practical algorithms for computable inference with this class of \glspl{HMM} and propose an automatic mixture pruning scheme that  results in linear computational costs, at the expense of some approximation error. 
The pruning approximations are based on the fact that although the number of components in the finite mixtures can grow rapidly, the number of components carrying most of the probability mass is seemingly stationary. See \Cref{fig:saturation}.
It thus seems natural to approximate the mixtures by, e.g., pruning all components with negligible mass, resulting in a roughly constant number of components. 
In Section \ref{sec: application} we discuss this and several pruning strategies.
\begin{figure}[t]
\begin{center}
\includegraphics[width = .8\textwidth]{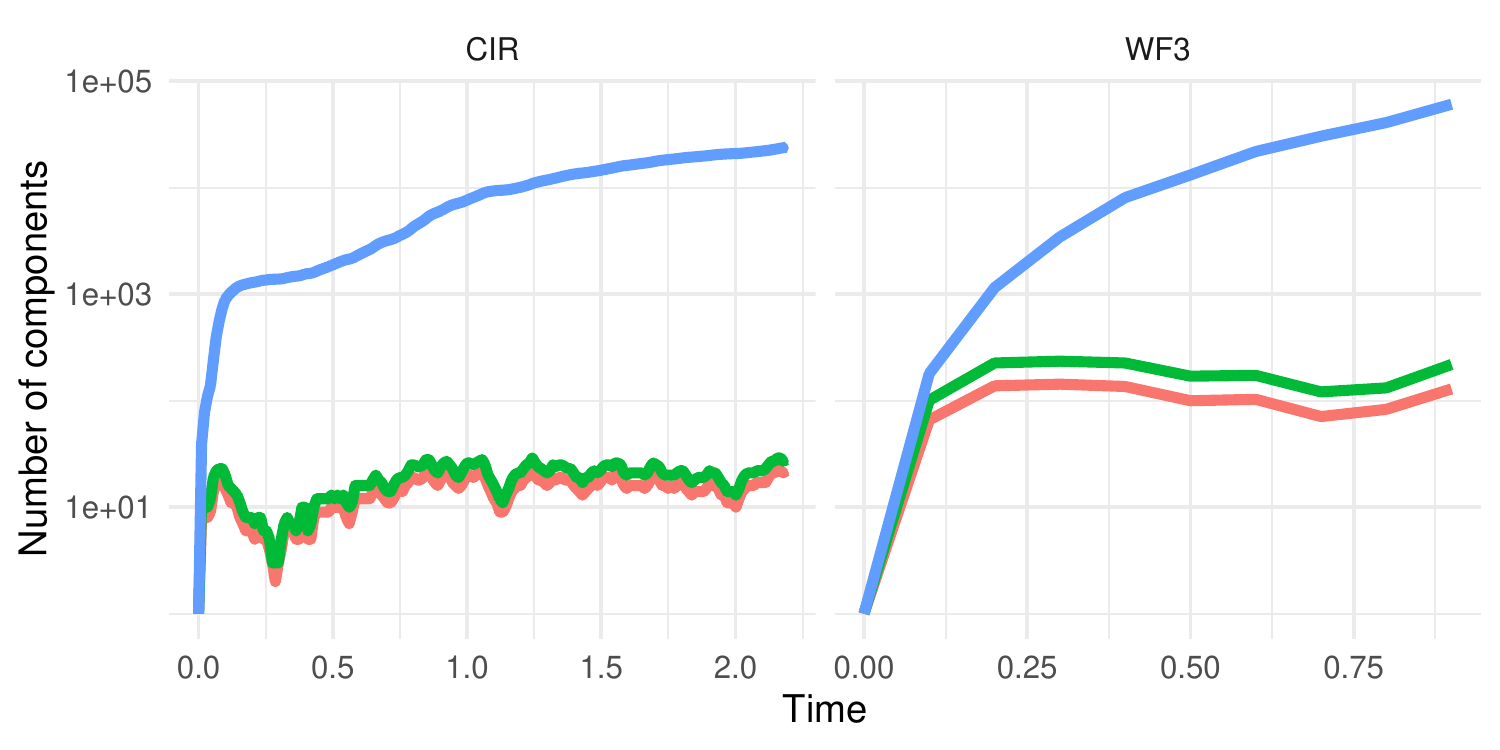}
\begin{quote}
\caption{ Number of components (in log scale) in the filtering densities as a function of time for two models illustrated in detail later (a Cox--Ingersoll--Ross process and a 3-component Wright--Fisher process), needed to account for 95\% (red), 99\% (green) and 100\% (blue) of the mass.\label{fig:saturation}}
\end{quote}
\end{center}
\end{figure}

By devising appropriate metrics for evaluating the quality of the mixture approximations, we compare the performance of our schemes against suitable particle filters and find superior performance. 
\Cref{likelihood_performance_teaser} provides a glimpse into these results, showing that for the two classes of models used in \Cref{fig:saturation}, computing the likelihood with one of the proposed approximation strategies outperforms particle filters both in accuracy and in computational efficiency.
\begin{figure}[t!]
\begin{center}
\includegraphics[width = .7\textwidth]{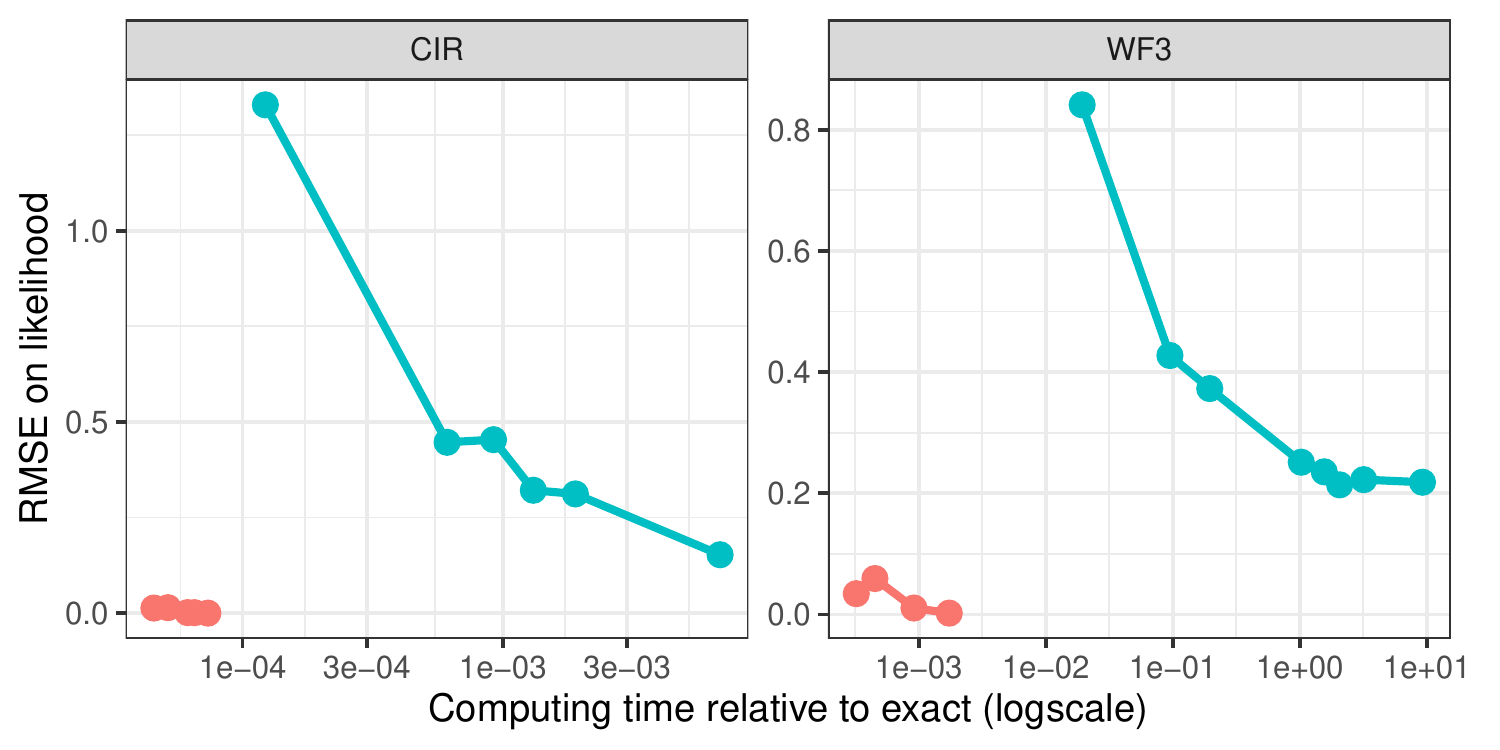}
\end{center}
\begin{quote}
\caption{ 
Root mean squared error for computing the likelihood, as a function of computing time, for one approximation strategy proposed later (red dots, for various levels of approximation) compared with using a particle filter (blue dots, for various number of particles), for the models in \Cref{fig:saturation}.
The computing time is relative to the time needed to compute the likelihood exactly (feasible for these model using the results presented later).  
\label{likelihood_performance_teaser}}
\end{quote}
\end{figure}

The code for performing optimal filtering, smoothing and parameter inference using the result presented here is made available as the Julia package  \texttt{DualOptimalFiltering}\footnote{Available at \url{https://github.com/konkam/DualOptimalFiltering.jl}.}. 
We also implemented a novel particle filtering algorithm using the recent exact sampling method of \cite{jenkins2017exact}, with the algorithm for exact sampling made available as the Julia package  \texttt{ExactWrightFisher}\footnote{Available at \url{https://github.com/konkam/ExactWrightFisher.jl}.}.

The rest of the paper is organised as follows.  
\Cref{sec: results} presents our theoretical results in terms of sufficient conditions for computable inference, details the associated general recursive formulae for updating the distributions of interest, and illustrates how this is concretely done in the case of a \gls{CIR} and a 3-component \gls{WF} signal, by describing explicitly the parameter updates for these models.
\Cref{sec:joint_inference} shows how to sample from the joint smoothing distribution in order to generate trajectories of the signal conditional on the data, and illustrates the strategy for the two above models. While for the first this comes out quite naturally from the computation, for the WF it requires a non trivial argument, which leverages on recent results by \cite{jenkins2017exact}.
\Cref{sec: application} discusses the algorithmic implementation and proposes some acceleration strategies via informed pruning of the mixtures, whereas  \Cref{sec: numerical experiments} presents numerical experiments concerning signal recovery, parameter estimation via maximum likelihood and full Bayesian inference via \gls{MCMC} algorithms. 
We compare the performance of our strategies against particle filtering. 
Concluding remarks are provided in a brief discussion. The most relevant proofs are included in the Appendix, while a few additional results are deferred to the Supplementary Material.


\section{Exact recursive inference for a class of HMMs}\label{sec: results}

\subsection{Setting and inferential goals}\label{sec: setting}

Let the hidden signal $X$ be a continuous-time Markov process on $\X\subseteq\R^{K}$ with initial distribution $\nu_{0}$ and transition density $P_{t}$. Let also $f_{x_{t}}(\cdot)$ be the \emph{emission densities}, i.e., the conditional densities of $Y_{t}$ given $X_t=x_t$, with a common dominating measure on $\Y$ for all $x_{t}\in \X$. 
Denote by $x_{0:T}:=(x_{0},\ldots,x_{T})$ the signal values at discrete times $0=t_{0}<\ldots<t_{T}$ and by $y_{0:T}:=(y_{0},\ldots,y_{T})$ the observations collected at those times. For notational simplicity, and without loss of generality, we assume only one data point is observed at each time, and that the observation times are equally spaced by an interval $\Delta=t_{i}-t_{i-1}$. We therefore denote by $P_{\Delta}$ the transition density of the signal on such intervals, write $x_{i}:=x_{t_{i}}$ and $ y_{i}:=y_{t_{i}}$, and (with some abuse) refer to $t_{i}$ as time $i$.

Inference in this setting is concerned with evaluating various conditional distributions of the signal given the observations. These are typically classified into: \emph{predictive densities}
\begin{equation}\label{predictive general}
\nu_{i|0:i-1}(x_{i}):=p(x_{i}\mid y_{0:i-1}) =\int_{\X} p(x_{i-1}\mid y_{0:i-1})P_{\Delta}(x_{i}\mid x_{i-1})\d x_{i-1}
\end{equation} 
which evaluate the law of the signal given past data;
\emph{filtering densities}
\begin{equation}\label{filtering distr general}
\nu_{i|0: i}(x_{i}):=p(x_{i}| y_{0:i})\propto \int_{\X^{i}} p(x_{0:i}, y_{0:i}) \d x_{0:i-1}, 
\end{equation}
given past and present data, where $p(x_{0:i}, y_{0:i})$ is the joint density of signal and observations; 
and \emph{smoothing densities}
\begin{equation}\label{smoothing distr general}
 \nu_{i|0:T}(x_{i}):=p(x_{i}\mid y_{0:T}) \propto \int_{\X^{T}} p(x_{0:T}, y_{0:T}) \d x_{0:i-1}\d x_{i+1:T}
\end{equation}
which aim at improving previous estimates of the signal once the whole dataset is available. 
A further quantity of primary interest is the likelihood of the observations
\begin{equation}\label{likelihood}
 p(y_{0:T}) = \int_{\X^{T+1}} p(x_{0:T}, y_{0:T}) \d x_{0:T},
\end{equation}
obtained by marginalising out the trajectory of the signal from the joint density in the integral.

Computable inference, as introduced in \cite{chaleyat2006computable} (cf.~Introduction), requires to characterise the above distributions, for the model at hand, as finite mixtures of densities, possibly of different dimensions for different collection times.
The following section identifies a set of sufficient conditions that guarantee such characterisations are available.

\subsection{Sufficient conditions for computable inference}\label{sec: assumptions}

We identify a set of sufficient conditions under which \eqref{predictive general}, \eqref{filtering distr general} and \eqref{smoothing distr general} can be written as finite mixtures of densities belonging to a given parametric family, and \eqref{likelihood} as a finite product of mixtures. Assumptions 1-3 below are the same as in \cite{Papaspiliopoulos2014a}, who studied computable filtering, while Assumption 4 is needed for computable smoothing.

\smallskip
\begin{assumption}[Reversibility]\label{A: reversibility}
The signal $X_{t}$ is reversible with respect to $\pi$, i.e., the detailed balance condition $\pi(x) P_{\Delta}(x'\mid x) = \pi(x') P_{\Delta}(x\mid x')$ holds.
\end{assumption}
\smallskip

Here  $\pi$ need not be a normalised density. In fact, our results carry over to signals with a sigma-finite reversible measure, as long as the distributions in \Cref{A: conjugacy} below can be normalised.

Define now, for $K\ge1$, the space of positive integer vectors and associated norm
\begin{equation}
 \label{eq:M}
 \M :=\Z_{+}^K= \{\,\mm=(m_1,\ldots,m_K): m_j \in \Z_{+}\}, \quad 
 \norm{\mm} =\sum\nolimits_{i=1}^{K}m_{i},
\end{equation}
where $\oo$ denotes the vector of zeros and $\ee_i=(\delta_{ij})_{j\ge1}$ is the canonical unit vector in the $i$th direction. We write $\mm \leq \nn$ if and only if $m_j \leq n_j$ for all $j$, and $\mm-\nn$ for the vector with $j$th element $m_j-n_j$. We will also need the grid of integer vectors lying below points in $\MM \subset \M$, denoted by
\begin{equation}\label{below Lambda}
\B(\MM) = \{\nn: \nn \leq \mm\,, \mm \in \MM\},
\end{equation}
and with a little abuse of notation we also let $\B(\mm):=\B(\{\mm\})=\{\nn: \nn \leq \mm\}$.

The second assumption provides key requirements on the existence of a certain \emph{dual process} for the signal (a definition is given below). 
To this end, let, for $\Th \subseteq \R^l$, $l \in \Z_{+}$, the function $\r:\Th \to \Th$ be such that the differential equation 
\begin{equation}
 \label{ODE}
 \d \Theta_t / \d t = \r(\Theta_t),\quad \Theta_0 = \theta_0,
\end{equation}
has a unique solution, denoted $\Theta_{t}$, for all $\theta_0$. 
Let
$q: \Z_{+} \to \R_+$ be an increasing function, $\rho:\Th \to
\R_+$ be a continuous function, and let $(M_t,\Theta_t)$ be a two-component process on $\M\times \Th$, where $\Theta_t$ evolves autonomously according to \eqref{ODE}, and when
$(M_t,\Theta_t)=(\mm,\theta)$, the process jumps to
$(\mm-\ee_j,\theta)$ at rate $\rho(\theta) q(\norm{\mm} ) m_j$.
I.e., $M_{t}$ is a non-homogeneous pure-death process on $\M$, and $\Theta_{t}$ a deterministic process that modulates the jump rates of $M_{t}$. The transition probabilities of $M_{t}$, denoted
\begin{equation}\label{transition probabilities}
p_{\mm, \nn}(t; \theta):=\P(M_{t}=\nn\mid M_{0}=\mm,\Theta_{0}=\theta),\quad \nn\le \mm,
\end{equation} 
and zero otherwise, are fully described in the Supplementary Material. 

\smallskip
\begin{assumption}[Duality] \label{A: duality}
There exists a family of functions $h:\X \times \M \times \Th \to \R_+$ with $h(x,\oo,\theta_{0})=1$ for some $\theta_{0}\in \Th$ and $\sup_x
h(x,\mm,\theta) < \infty$ for all $(\mm,\theta)$, such that
\begin{equation}\label{duality identity}
 \E[h(X_t,\mm,\theta)\mid X_{0}=x] = \E[h(x,M_t,\Theta_t)\mid (M_0,\Theta_0)=(\mm,\theta)].
\end{equation}

\end{assumption}
\medskip

When \eqref{duality identity} holds, $(M_t,\Theta_t)$ is said to be \emph{dual} to $X_t$ with respect to functions $h$, called \emph{duality functions}. See \cite{jansen2014notion}.
The conditional expectations in \eqref{duality identity} are taken with respect to the law of $X_{t}$ on the left hand side, and with respect to that of $(M_t,\Theta_t)$ on the right hand side. 
When $K=0$ or $l=0$, the dual process is just $\Theta_t$ or $M_t$ respectively, and we adopt the convention that $\rho(\theta)\equiv1$ whenever $l=0$.

The third assumption essentially amounts to what in Bayesian statistics is known as \emph{conjugacy}.

\smallskip
\begin{assumption}[Conjugacy] \label{A: conjugacy}
For $\pi$ as in \Cref{A: reversibility} and $h$ as in \Cref{A: duality}, the emission density $f_{x_{t}}(\cdot)$ is conjugate to densities in the parametric family
\begin{equation}\label{family F}
\F=\{g(x,\mm,\theta)=h(x,\mm,\theta) \pi(x),\, \mm \in \M, \theta \in \Th\},
\end{equation} 
i.e., there exist an increasing function $t:\Y \times \M \to \M$ and a function $T: \Y \times \Th \to \Th$ such that if $X\sim g(x,\mm,\theta)$ and $Y\mid (X=x)\sim f_{x}(y)$, then $X\mid (Y=y)\sim g(x,t(y,\mm),T(y,\theta))$.
\end{assumption}
\smallskip

Here $g(x,\oo,\theta_{0})=h(x,\oo,\theta_{0})\pi(x)=\pi(x)$ is identified with the prior distribution (cf.~\Cref{A: duality}), and $g(x,\mm,\theta)$ is the posterior distribution of $X_{t}$ given $y\sim f_{x_{t}}(\cdot)$ if $t(y,\oo)=\mm$ and $T(y,\theta_{0})=\theta$. 

The fourth assumption is needed for obtaining the smoothing densities in computable form.

\smallskip
\begin{assumption}[$h$-stability]\label{A: h-stability}
For $h$ as in \Cref{A: duality}, there exist functions $d:\M^{2}\to\M$ and $e:\Th^{2}\to\Th$ such that for all $x\in\X,\ \mm,\mm' \in\M,\ \theta, \theta' \in \Th$
\begin{equation}\label{hh_stab}
 h(x, \mm, \theta)h(x, \mm', \theta') 
 =C_{\mm, \mm', \theta, \theta'} h(x, d(\mm, \mm'), e(\theta, \theta')),
\end{equation}
where $C_{\mm, \mm', \theta, \theta'}$ is constant in $x$.
\end{assumption}
\smallskip

At close inspection, \Cref{A: h-stability} may appear to follow from the conjugacy in \Cref{A: conjugacy}. In fact, when two observations $Y,Y'$ independently give  posteriors $ g(x, \mm, \theta)$ and $g(x, \mm', \theta')$, then $g(x, d(\mm, \mm'), e(\theta, \theta'))$ is simply the posterior obtained by using data $(Y,Y')$ jointly, and $C_{\mm, \mm', \theta, \theta'}$ simply provides a reparameterization. However, it appears not immediate that this argument is valid for any value of $\theta$, so we state it as an assumption.


\subsection{Recursive formulae}\label{sec: main results}

In this section we derive our theoretical results, and present general recursive formulae for updating the computable representations of 
\eqref{predictive general}-\eqref{filtering distr general}-\eqref{smoothing distr general}-\eqref{likelihood}.

\subsubsection{Filtering and prediction}\label{sec: filtering-prediction}

Denote by
\begin{equation}\label{marginals}
\begin{aligned}
\mu_{\mm,\theta}(y):=\int_{\X}f_{x}(y)g(x, \mm, \theta)\d x
\end{aligned}
\end{equation} 
the marginal density of $Y\sim f_{x}(\cdot)$ when $X\sim g(x, \mm, \theta)$, for $g$ in \eqref{family F}. For $\Theta_{t}$ as in \eqref{ODE}, $t(\cdot,\cdot),T(\cdot,\cdot)$ as in \Cref{A: conjugacy} and $i=0,\ldots,T$, define also the quantities
\begin{equation}
\begin{aligned}
\label{vartheta and M-sets}
\vartheta_{i|0:i}:=&\,T(y_{i}, \vartheta_{i|0:i-1}),
\quad \quad \ \
\vartheta_{i|0:i-1} := \Theta_{\Delta}(\vartheta_{i-1|0:i-1}),
\quad \quad 
\vartheta_{0|0:-1} := \theta_0\\
\MM_{i|0:i}:=&\,t(y_{i},\MM_{i|0:i-1}), \quad \quad
\MM_{i|0:i-1}:= \B(\MM_{i-1|0:i-1}),\quad \quad 
\MM_{0|0:-1} := \{\oo\}.
\end{aligned}
\end{equation} 
Here, $\vartheta_{i|0:i-1}$ denotes the state of the deterministic component of the dual process at time $i$, after the propagation from time $i-1$ and before updating with the datum collected at time $i$, and $\vartheta_{i|0:i}$ the state after such update. Similarly, $\MM_{i|0:i-1}$ denotes the set, at time $i$ before the update, of what are called here \emph{active indices}, i.e., the points in $\M$ identifying mixture components with strictly positive weight, and $\MM_{i|0:i}$ those active after the update. 
Note that the quantities in \eqref{vartheta and M-sets} are deterministic and can be computed on the basis of the dataset $y_{0:T}$, by means of the update functions $t,T$ in Assumption \ref{A: conjugacy}, of the solution of \eqref{ODE} and of \eqref{below Lambda}.

The following Theorem provides the recursive formulae for prediction and filtering in \eqref{predictive general} and \eqref{filtering distr general}.

\begin{theorem}
\label{prop:rec_filtering}
Let \Cref{A: reversibility} to \ref{A: conjugacy} hold, and let 
\begin{equation}\nonumber
\sum_{\mm \in \MM_{i-1|0:i-1}}w_\mm^{(i-1)}g(x, \mm, \vartheta_{i-1|0:i-1})
\end{equation} 
 be the density of $x_{i-1}$ conditional on $y_{0:i-1}$. Then  \eqref{predictive general} and \eqref{filtering distr general} are the finite mixtures of densities
\begin{align}
 \nu_{i|0:i-1}(x) 
 =&\, \sum_{\mm \in \MM_{i|0:i-1}}w_\mm^{(i-1)'}g(x, \mm, \vartheta_{i|0:i-1}),\label{prediction in thm}\\
 \nu_{i|0:i}(x)
=&\, \sum_{\mm \in \MM_{i|0:i}}w_\mm^{(i)} g(x, \mm, \vartheta_{i|0:i})
\label{filtering in thm}
\end{align}
where, for $p_{ \nn, \mm}(\Delta; \vartheta_{i|0:i})$ and $\mu_{\mm,\theta}$ as in \eqref{transition probabilities} and \eqref{marginals}, the mixture weights $ w_\mm^{(i-1)'},w_{\mm}^{(i)}$ are given by
\begin{equation}\label{weights of filtering recursion}
\begin{aligned}
 w_\mm^{(i-1)'} = &\, \sum_{\nn \in \MM_{i-1|0:i-1}:\ \nn \ge \mm}w_{\nn}^{(i-1)}p_{ \nn, \mm}(\Delta; \vartheta_{i-1|0:i-1}), \quad \mm \in \MM_{i|0:i-1},\\
w_{\mm}^{(i)} \propto &\, 
\mu_{\nn, \vartheta_{i|0:i-1}}(y_{i})w_{\nn}^{(i-1)'}, 
\quad \mm =t(y_{i},\nn),\nn\in \MM_{i|0:i-1},
\end{aligned}
\end{equation} 
and 0 elsewhere. 
\end{theorem}

This result is due to \cite{Papaspiliopoulos2014a} and included here in the present notation for ease of reference. Note that a predictive distribution for the signal at time $T+t$ for arbitrary $t>0$ can be easily obtained from \eqref{prediction in thm} by letting $\Delta=t$ in the transition probabilities $p_{ \nn, \mm}(\Delta; \vartheta_{i|0:i})$ (cf.~Supplementary Material).

\subsubsection{Likelihood}\label{sec: likelihood}

From \eqref{likelihood}, we can write the likelihood as
\begin{equation}
\begin{aligned}
p(y_{0:T}) 
=&\, \int_{\X^{T+1}} \nu(x_0)\prod_{i=0}^{T}f_{x_{i}}(y_{i})\prod_{i=1}^T P_{\Delta_{i}}(x_{i}\mid x_{i-1})\d x_{0:T}
\end{aligned}
\end{equation}
by integrating the signal trajectory out of $p(y_{0:T}\mid x_{0:T})$ and using the conditional independence of $y_{i}$ given $x_{i}$.
Alternatively, writing
\begin{equation}\label{eq:lik_markov} 
\begin{aligned}
p(y_{0:T}) =&\,p(y_{0})\prod_{i=1}^{T}p(y_{i}\mid y_{0:i-1})
=\int_{\X} f_{x_{0}}(y_{0})\nu(x_0)\prod_{i=1}^{T}\int_{\X}f_{x_{i}}(y_{i})p(x_{i}\mid y_{0:i-1}),
\end{aligned}
\end{equation} 
highlights the dependence on the predictive densities \eqref{predictive general}.

The following theorem shows that in the present setting, the above expressions are finite products of finite mixtures of distributions in \eqref{family F}, marginalised over the hidden state as in \eqref{marginals}. 

\begin{theorem}\label{thm: likelihood}
Under the assumptions of \Cref{prop:rec_filtering}, setting $\nu_{0}=\pi$, we have
 \begin{align}\label{full likelihood}
p(y_{0:T}) =\mu_{\oo,\theta_0}(y_0)\prod_{i = 1}^{T}\sum_{\nn \in \MM_{i|0:i-1}}w_\nn^{(i-1)'} \mu_{\nn, \vartheta_{i|0:i-1}}(y_{i})
\end{align}
with $\mu_{\nn,\vartheta_{i|0:i-1}}$ as in \eqref{marginals}, $\vartheta_{i|0:i-1},\MM_{i|0:i-1}$  as in \eqref{vartheta and M-sets} and $w_\nn^{(i-1)'}$ as in \eqref{weights of filtering recursion}.
\end{theorem}

The proof of \Cref{thm: likelihood} is provided in the Appendix. 

\subsubsection{Smoothing} \label{marginal_smoothing}

Let $0\le i\le T-1$. 
Bayes' Theorem and conditional independence allow to write  \eqref{smoothing distr general} as
\begin{equation}\label{smoothing_expression}
p(x_i\mid y_{0:T}) 
\propto p(y_{i+1:T}\mid x_i)p(x_i\mid y_{0:i})
\end{equation} 
where the right hand side involves the filtering distribution $ p(x_i\mid y_{0:i})$, available from \Cref{prop:rec_filtering}, and the  \emph{cost-to-go function} $p(y_{i+1:T}\mid x_i)$, sometimes called information filter, which is the likelihood of future observations given the signal. 

Denote by $\avt_{i}, \avt_{i}',\aMM_{i}, \aMM_{i}'$ the quantities defined in 
\eqref{vartheta and M-sets} computed backwards. Equivalently, these are computed as in \eqref{vartheta and M-sets} with data in reverse order, i.e.~using $y_{T:0}$ in place of $y_{0:T}$, namely
\begin{equation}\label{backward lambda e vartheta}
\begin{aligned}
\avt_{i|i+1:T}=&\,\Theta_{\Delta}(\avt_{i+1|i+1:T}), \quad \quad 
\avt_{i|i:T}=T(y_{i},\avt_{i|i+1:T}),\quad \quad 
\avt_{T|T}=T(y_{T},\theta_{0})\\
\aMM_{i|i+1:T}=&\,\B(\aMM_{i+1|i+1:T}), \quad \quad 
\aMM_{i|i:T}=t(y_{i},\aMM_{i|i+1:T}), \quad \quad 
\aMM_{T|T}=\{t(y_{T},\oo)\}.
\end{aligned}
\end{equation} 
The following proposition identifies an explicit expression for the cost-to-go function.

\begin{restatable}[]{proposition}{pfs}\label{prop: prediction for smoothing}
Let Assumptions \ref{A: reversibility}-\ref{A: duality} above hold. 
For all $0 \le i \le T-1$, we have
\begin{equation}\label{cost-to-go function}
 p(y_{i+1:T}\mid x_i) 
=
\sum_{\mm \in \aMM_{i|i+1:T} } \overleftarrow w^{(i+1)}_{\mm}h(x_{i}, \mm, \avt_{i|i+1:T})
\end{equation}
with 
\begin{equation}\label{weights of prediction for smoothing}
\overleftarrow w^{(i+1)}_{\mm}
=\sum_{\nn\in\aMM_{i+1|i+2:T}:\, t(y_{i+1}, \nn)\ge \mm}
\overleftarrow w_{\nn}^{(i+2)}
\mu_{\nn,\avt_{i+1|i+2:T}}(y_{i+1})
p_{t(y_{i+1},\nn),\mm}(\Delta; \avt_{i+1|i+1:T})
\end{equation} 
and $\avt_{i+1},\avt_{i|i+1:T},\aMM_{i+1|i+1:T},\aMM_{i|i+1:T}$ as in \eqref{backward lambda e vartheta}.
\end{restatable}
\noindent The proof of \Cref{prop: prediction for smoothing} is provided in the Appendix.

\vspace{1em}
An intuition on the above result can be obtained by considering that for $i=T-1$, $ p(y_{T}\mid x_{T-1}) $   simply averages over the possible values of the signal at time $T$, and computes the likelihood of $Y_{T}$ given such values. This would in general yield an infinite expansion based on the transition kernel of the signal. Here the duality relation in Assumption \ref{A: duality} allows to express this quantity as a finite linear combination of duality functions, which are in turn ratios of likelihoods over marginal likelihoods. 

The following Theorem shows that under the stated assumptions, the smoothing density \eqref{smoothing_expression} as well takes the form of a finite mixture of densities in \eqref{family F}.

\begin{theorem}\label{thm: smoothing}
Let \Cref{A: reversibility} to \ref{A: h-stability} above hold and let $\nu_{0}=\pi$. 
Then, for $0 \le i \le T-1$,
 \begin{equation}
 p(x_{i}\mid y_{0:T})=
\sum_{\mm \in \aMM_{i|i+1:T},\ \nn \in \MM_{i|0:i}}
w_{\mm,\nn}^{(i)}g(x_{i}, d(\mm,\nn), e(\avt_{i|i+1:T},\vartheta_{i|0:i})),
 \end{equation}
 with 
 \begin{equation}
w_{\mm,\nn}^{(i)}
\propto\overleftarrow w^{(i+1)}_{\mm} w_\nn^{(i)}C_{\mm, \nn, \avt_{i|i+1:T},\vartheta_{i|0:i}},
\end{equation} 
 $w_\nn^{(i)}$ as in \eqref{weights of filtering recursion}, $\overleftarrow w^{(i+1)}_{\mm}$ as in \eqref{weights of prediction for smoothing} and $C_{\mm, \nn, \avt_{i|i+1:T},\vartheta_{i|0:i}}$ as in \eqref{hh_stab}.
\end{theorem}
\noindent The proof of \Cref{thm: smoothing} is provided in the Appendix.

\vspace{1em}
Here the intuition is instead that the distribution of the signal at time $i$ is evaluated by appropriately ``interpolating'' its distribution at adjacent times $i-1,i+1$, given data $y_{0:i-1}$ and $y_{i+1:T}$ respectively. This interpolation is performed in principle using again the transition kernel of the signal, but it is again the duality of Assumption \ref{A: duality} that allows to reduce the resulting expression from a doubly infinite series to a double finite sum.
Here both $\MM_{i|0:i}$ and $\aMM_{i|i+1:T}$, which are sets of  active indices at time $i$, have finite cardinality, hence the smoothing densities can be computed recursively with a finite number of operations.

\subsection{Illustration}\label{sec: illustration}

\glsreset{CIR}
\glsreset{WF}

We illustrate the above results for two \glspl{HMM} of interest that fall in our setting: a one dimensional signal driven by a \gls{CIR} diffusion (which is also a continuous-state branching process), with Poisson distributed observations; and a signal driven by a \gls{WF} diffusion on the $(K-1)$-dimensional simplex with categorical observations.
Note that our results include also finite state space models and linear Gaussian systems (cf.~\citealp{Papaspiliopoulos2014a}), whose details for filtering and smoothing are well known and omitted here.

\subsubsection{Cox--Ingersoll--Ross signals}\label{sec: illustration_CIR}

The \gls{CIR} process 
is the solution of the stochastic differential equation
\begin{equation}\label{CIR SDE}
 \d X_{t} = (\delta\sigma^2 - 2\gamma X_{t})\d t + 2\sigma\sqrt{X_{t}}\d B_t, \quad \quad 
 \delta, \gamma, \sigma > 0,
\end{equation}
whose stationary distribution is the gamma density $\text{Ga}(\delta/2, \gamma/\sigma^2)$ with shape $\delta/2$ and rate $\gamma/\sigma^2$. 
A conjugate emission density $f_{x}(y)$ is the Poisson distribution $\text{Po}(\lambda x)$ with mean $\lambda x$. 
Choosing as prior $\pi(x)=g(x,0,\theta_{0})=\text{Ga}(x;\delta/2,\gamma/\sigma^2)$, with $\theta_0=\gamma / \sigma^{2}$, and letting $f_{x}(y)$ as above, the update of $g(x,0,\theta_{0})$ given one observation  $y$ yields $g(x,m,\theta)=g(x,t(y,0),T(y,\theta_{0}))=\text{Ga}(\delta/2+y,\gamma/\sigma^2+\lambda)$. 
Thus in this case
\begin{equation}\label{CIR h function}
\begin{aligned}
h(x, m, \theta)=\frac{\Gamma(\delta / 2)}{\Gamma(\delta / 2+m)}\left(\frac{\gamma}{\sigma^{2}}\right)^{-\delta / 2} \theta^{\delta / 2+m} x^{m} \exp \left\{-\left(\theta-\gamma / \sigma^{2}\right) x\right\}
\end{aligned}
\end{equation} 
together with $t(y,m)=m+y$, $T(y, \theta) =\theta + \lambda$. 
It can be easily verified that the marginals are Negative-Binomial densities, e.g.~for $y$ as above
\begin{equation}
 \mu_{m, \theta}(y) = 
 \NBin\left(y; \frac{\delta}{2}+m, \frac{\theta}{\theta+\lambda}\right).
\end{equation}
Furthermore, \eqref{CIR SDE} has dual given by a non-homogeneous death process on $\Z_{+}$ with transition probabilities 
 \begin{equation}\label{determin dual CIR}
 p_{m, m-i}(\Delta;\theta) = \text{Bin}(m-i; m, \Theta_{\Delta}(\theta)),\quad \quad 
 \Theta_{\Delta}(\theta) = \frac{(\gamma/\sigma^{2})\theta \mathrm{e}^{2 \gamma \Delta}}{\theta \mathrm{e}^{2 \gamma \Delta}+\gamma / \sigma^{2}-\theta},
\end{equation} 
where $\theta$ in \eqref{determin dual CIR} is the last available value of the gamma rate parameter. 

In this model, all filtering and predictive densities are finite mixtures of gamma densities. Specifically, let $w_{n}$ be the weights of the predictive  density at time $i$, given by the expression
\begin{equation}\label{CIR previous predictive}
\nu_{i|0:i-1}(x_{i})=\sum_{0\le n\le N_{i-1}}w_{n}\Ga(x;\delta/2+n,\vartheta_{i|0:i-1}), 
\quad \quad N_{i-1}=\sum_{j=0}^{i-1}y_{j}.
\end{equation} 
Here $\vartheta_{i|0:i-1}$ is obtained as in  \eqref{vartheta and M-sets} by recursively computing $\vartheta_{0|0}=T(y_{0},\gamma/\sigma^{2})$, $\vartheta_{1|0}=\Theta_{\Delta}(\vartheta_{0|0})$, $\vartheta_{1|0:1}=T(y_{1},\vartheta_{1|0})$, and so on.
The marginal density of a single observation $y_{i}$ is obtained by integrating the emission density $\text{Po}(\lambda x_{i})$ with respect to \eqref{CIR previous predictive}, yielding
\begin{equation}\label{CIR marginal}
\mu_{\nu_{i|0:i-1}}(y_{i})
= \int_{\X}\text{Po}(y_i; \lambda x_i)\nu_{i|0:i-1}(x_{i}) = \sum_{0\le n\le N_{i-1}}w_{n}
 \NBin\left(y_{i}; \frac{\delta}{2}+n, \frac{\vartheta_{i|0:i-1}}{\vartheta_{i|0:i}}\right),
\end{equation} 
with $N_{i-1}$ and $w_{n}$ as in \eqref{CIR previous predictive}. The density of all observations \eqref{full likelihood} is thus a product of  mixtures as in \eqref{CIR marginal}, with weights computed recursively as described in Theorem \ref{thm: likelihood}.
Upon observing $y_{i}$, the filtering density at time $i$ reads
\begin{equation}\label{CIR filtering}
\begin{aligned}
\nu_{i|0:i}(x_{i})
=&\,
\sum_{y_{i}\le m\le N_{i-1}+y_{i}}w_{m}^{(i)}\Ga(x_{i};\delta/2+m,\vartheta_{i|0:i})\\
w_{m}^{(i)} \propto &\, 
 w_{n}
  \NBin\left(y_{i}; \frac{\delta}{2}+n, \frac{\vartheta_{i|0:i-1}}{\vartheta_{i|0:i}}\right), \quad m=n+y_{i}, \quad 0\le n\le N_{i-1},
\end{aligned}
\end{equation} 
with $w_{n}$ given in \eqref{CIR previous predictive}. Then, the predictive density for $x_{i+1}$ given $y_{0:i}$ is 
\begin{equation}\label{CIR prediction}
\begin{aligned}
\nu_{i+1|0:i}(x_{i+1})=&\,\sum_{0\le n\le N_{i}}w_{n}^{(i)'}\Ga(x_{i+1};\delta/2+n,\vartheta_{i+1|0:i}), \quad \quad N_{i}=N_{i-1}+y_{i}\\
w_{n}^{(i)'} = &\, \sum_{n\le \ell\le N_{i}}w_{\ell}^{(i)}p_{ \ell,n}(\Delta;  \vartheta_{i|0:i}), 
\end{aligned}
\end{equation} 
with $p_{\ell,m}$  and  $\Theta_{\Delta}(\theta)$ as in \eqref{determin dual CIR} and  $w_{\ell}^{(i)}$ as in \eqref{CIR filtering}.

Additionally, \Cref{A: h-stability} is satisfied with $ d(m_1, m_2) =m_1 + m_2, 
 e(\theta_1, \theta_2) =\theta_1 + \theta_2 - \gamma/\sigma^2$ 
and
\begin{align}
 C_{m_1, m_2, \theta_1, \theta_2} =&\, \Gamma(\delta/2)\bigg(\frac{\gamma}{\sigma^{2}}\bigg)^{-\delta/2}\frac{\Gamma(\frac{\delta}{2}+m_1+m_2)}{\Gamma(\frac{\delta}{2}+m_1)\Gamma(\frac{\delta}{2}+m_2)}\frac{(\theta_1)^{\delta/2+m_1}(\theta_2)^{\delta/2+m_2}}{(\theta_1+\theta_2-\gamma/\sigma^2)^{\delta/2+m_1+m_2}}.
 \end{align}
 It is easy to show that, with $\theta_0 = \gamma/\sigma^2$, then $\theta_1,\theta_2,e(\theta_1, \theta_2)\ge\gamma/\sigma^2$.
Hence, the marginal smoothing density of $x_{i}$ given $y_{0:T}$ can be obtained by combining the filtering density $\nu_{i|0:i}(x_{i})$ and the cost-to-go function $p(y_{i+1:T}|x_{i})$. 
The latter is
\begin{equation}\label{CIR cost-to-go function}
\begin{aligned}
 p(y_{i+1:T}\mid x_i) 
=&\,
\sum_{0\le m \le M_{i+1}} \overleftarrow w^{(i+1)}_{m}h(x_{i}, m, \vartheta'), \quad \quad M_{i+1}=\sum_{j=i+1}^{T}y_{i},
\end{aligned}
\end{equation}
with $h$ as in \eqref{CIR h function} and $\vartheta'=\avt_{i|i+1:T}$ as in \eqref{backward lambda e vartheta}. 
The weights in \eqref{CIR cost-to-go function} are obtained from those of $p(y_{i+2:T}|x_{i+1})$, denoted $\overleftarrow w_{n}^{(i+2)}$, as
\begin{equation}\nonumber
\overleftarrow w^{(i+1)}_{m}
=
\sum_{n:\, n+y_{i+1}\ge m}
\overleftarrow w_{n}^{(i+2)}
 \NBin\left(y_{i+1}; \frac{\delta}{2}+n, \frac{\vartheta''}{\vartheta''+\lambda}\right)
\text{Bin}(m; n+y_{i+1}, \vartheta')
\end{equation} 
with $\vartheta''=\avt_{i+1|i+2:T}$. 
The smoothing density therefore is the finite mixture of Gamma densities
  \begin{equation}
\begin{aligned}
 p(x_{i}\mid y_{0:T})
 =&\,
 \sum_{y_i\le n\le N_{i}} 
\sum_{0\le m\le M_{i+1}}
w_{m,n}^{(i)}\Ga\bigg(x_{i}; \frac{\delta}{2}+m+n, \vartheta'+\vartheta-\gamma/\sigma^{2}\bigg)\\
\end{aligned}
 \end{equation}
with
 \begin{equation}
w_{m,n}^{(i)}\propto
\overleftarrow w^{(i+1)}_{m} w_{n}^{(i)}C_{m,\ell, \vartheta',\vartheta}
\end{equation} 
where $\overleftarrow w^{(i+1)}_{m}$ is as in \eqref{CIR cost-to-go function} and $w_{n}^{(i)}$ as in \eqref{CIR filtering}.


\subsubsection{Wright--Fisher signals}\label{sec: WF illustr}

The $K$-component \gls{WF} model is a diffusion process taking values in the simplex of nonnegative vectors $\xx=(x_{1},\ldots,x_{K})$ whose coordinates sum up to one (later simply called the simplex). 
It is characterised by its infinitesimal operator is 
\begin{equation}\label{WF operator}
 \mathcal{A} = \frac{1}{2}\sum_{i=1}^K(\alpha_i - |\aa| x_j)\frac{\partial}{\partial x_i} + \frac{1}{2}\sum_{i,j=1}^Kx_i(\delta_{ij}-x_j)\frac{\partial^2}{\partial x_i\partial x_j}, \quad 
 |\aa|=\sum_{j=1}^{K}\alpha_{j},
\end{equation}
where $\alpha_{j}>0$ for all $j$, and the domain of $\mathcal{A}$ can be taken to be the class of twice differentiable functions on the simplex. 
This diffusion is stationary and reversible with respect to the Dirichlet distribution whose density  on the simplex with respect to the Lebesgue measure is proportional to $\prod_{j=1}^{K}x_{j}^{\alpha_{j}}$. See \cite{ethier1986markov}.

A distribution conjugate to the Dirichlet, seen as prior density for a simplex-valued variable, is the Multinomial distribution, denoted here $\text{MN}(\yy;\norm{\yy},\xx)$, where $x_{j}$ is the probability of drawing category $j$ in a sample of size $\norm{\yy}$ and $\yy=(y_{1},\ldots,y_{K})$ are the multiplicities being drawn. 
Upon observing $\yy$, the Dirichlet density is updated by replacing $\aa$ with $\aa+\yy$, i.e., $\alpha_{j}$ with $\alpha_{j}+y_{j}$ for $j=1,\ldots,K$. 

Moreover, the \gls{WF} signal is known to be dual to a death process $D_{t}$ on $\Z_{+}^{K}$ that jumps from $\mm$ to $\mm-\ee_{j}$ at rate $m_{j}(|\aa|+|\mm|-1)/2$, with respect to functions 
\begin{equation}\label{hWF}
 h(\xx, \mm) = \frac{\Gamma(|\aa| +| \mm| )}{\Gamma(| \aa| )}\prod_{j=1}^K\frac{\Gamma(\alpha_j)}{\Gamma(\alpha_j+m_j)}x_{j}^{m_{j}}.
\end{equation}
This dual process has no deterministic component, and its transition probabilities $p_{\mm,\nn}(\Delta)$ are obtained by specialising Lemma S1.1 in the supplementary material to the case $\rho\equiv1$.
It can be easily verified that the marginal distribution of $\yy$ is a Dirichlet-Multinomial density, e.g., given parameter $\aa+\mm$, 
\begin{equation}\nonumber
\mu_{\mm}(\yy) = \DirMN(\yy;\aa+\mm)
:=\binom{\norm{\yy}}{\yy}\frac{\prod_{j=1}^K(\alpha_{j}+m_j)_{(y_{j})}}{(\norm{\aa+\mm})_{(\norm{\yy})}},
\end{equation} 
where $a_{(n)}=a(a+1)\cdots(a+n-1)$ is the Pochhammer symbol.
In this model, all filtering and predictive densities are finite mixtures of Dirichlet densities. Specifically, let the density of the signal at time $i$, before conditioning on $\yy_{i}$, be
\begin{equation}\label{WF previous predictive}
\nu_{i|0:i-1}(\xx_{i})=
\sum_{\oo\le \nn\le \textbf{N}_{i-1}}w_{\nn}
\Dir(\xx_{i};\aa+\nn), 
\quad \quad \textbf{N}_{i-1}=\sum_{j=0}^{i-1}\yy_{j}.
\end{equation} 
The marginal density of $\yy_{i}$ is obtained by integrating the Multinomial emission density with respect to \eqref{CIR previous predictive}, yielding
\begin{equation}\nonumber
\mu_{\nu_{i|0:i-1}}(\yy_{i})
= \int_{\X}\text{MN}(\yy_i; |\yy_i|, \xx)\nu_{i|0:i-1}(\xx_{i}) = \sum_{\oo\le \nn\le \textbf{N}_{i-1}}w_{\nn}
\DirMN(\yy;\aa+\nn),
\end{equation} 
with $\textbf{N}_{i-1}$ and $w_{\nn}$ as above, and the density of all observations \eqref{full likelihood} is thus a product of mixtures, with weights computed recursively.
Upon observing $\yy_{i}$, the filtering density at time $i$ reads
\begin{equation}\label{WF filtering}
\begin{aligned}
\nu_{i|0:i}(\xx_{i})=&\,\sum_{\yy_{i}\le \mm\le \textbf{N}_{i-1}+\yy_{i}}w_{\mm}^{(i)}
\Dir(\xx_{i};\aa+\mm)\\
w_{\mm}^{(i)} \propto &\, 
 w_{\nn}
\DirMN(\yy;\aa+\nn), \quad \mm=\nn+\yy_{i}, \quad \oo\le \nn\le \textbf{N}_{i-1},
\end{aligned}
\end{equation} 
and the predictive density for $\xx_{i+1}$ given $\yy_{0:i}$ is
\begin{equation}\label{WF predictive}
\begin{aligned}
\nu_{i+1|0:i}(\xx_{i+1})=&\,\sum_{\oo\le \nn\le \textbf{N}_{i}}w_{\nn}^{(i)'}
\Dir(\xx_{i+1};\aa+\nn), \quad \quad \textbf{N}_{i}=\textbf{N}_{i-1}+\yy_{i},\\
w_{\nn}^{(i)'} = &\, 
\sum_{\nn\le \ll\le \textbf{N}_{i}} w_{\ll}^{(i)}p_{\ll,\nn}(\Delta),
\end{aligned}
\end{equation} 
with $w_{\ll}^{(i)}$ as in \eqref{WF filtering}.
Additionally, \Cref{A: h-stability} for $h(\xx,\mm)$ is satisfied with $ d(\mm_1, \mm_2) =\mm_1 + \mm_2$ and
\begin{equation}\nonumber
 C_{\mm_1, \mm_2} = \frac{\Gamma(| \aa+\mm_1| )\Gamma(| \aa+\mm_2| )}{\Gamma(| \aa| )\Gamma(| \aa+\mm_1+\mm_2| )} \notag \prod_{j=1}^K\frac{\Gamma(\alpha_j)\Gamma(\alpha_j + m_{j1}+ m_{j2})}{\Gamma(\alpha_j + m_{j1})\Gamma(\alpha_j + m_{j2})},
\end{equation} 
and the cost-to-go function $p(\yy_{i+1:T}|\xx_{i})$ is
\begin{equation}\label{WF cost-to-go function}
\begin{aligned}
 p(\yy_{i+1:T}\mid \xx_i) 
=&\,
\sum_{\oo\le \mm \le \textbf{M}_{i+1}}
 \overleftarrow w^{(i+1)}_{\mm}h(\xx_{i}, \mm), \quad \quad 
 \textbf{M}_{i+1}=\sum_{j=i+1}^{T}\yy_{i},\\
\overleftarrow w^{(i+1)}_{\mm}
=&\,
\sum_{\nn:\, \nn+\yy_{i}\ge \mm}
\overleftarrow w_{\nn}^{(i+2)}
\textbf{DM}(\yy_{i+1};\aa+\nn)
p_{\nn+\yy_{i},\mm}(\Delta).
\end{aligned}
\end{equation}
Therefore the smoothing density is the finite mixture of Dirichlet densities
  \begin{equation}
 p(\xx_{i}\mid \yy_{0:T})=
 \sum_{\yy_i\le \nn\le \textbf{N}_{i-1}} \sum_{\oo\le \mm\le \textbf{M}_{i+1}}
w_{\mm,\nn}^{(i)}\Dir(\xx_{i}; \aa+\mm+\nn),
 \end{equation}
with
 \begin{equation}
w_{\mm,\nn}^{(i)}\propto
\overleftarrow w^{(i+1)}_{\mm} w_{\nn}^{(i)}C_{\mm,\nn}
\end{equation} 
where $\overleftarrow w^{(i+1)}_{m}$ is as in \eqref{WF cost-to-go function} and $w_{n}^{(i)}$ as in \eqref{WF filtering}.

%
%

\section{Simulation from the joint smoothing distribution}\label{sec:joint_inference}

The results presented in Section \ref{sec: results} can be elaborated to provide methods to simulate trajectories of the signal conditional on the entire dataset. These can be interpreted as samples from the posterior distribution of the signal, evaluated at the skeleton of observation times.
These draws in turn can be used to estimate the model parameters through a \gls{MCMC} strategy. 
More specifically, the availability of a method for sampling conditional trajectories opens the way to performing full inference on the model parameters, using a Metropolis-within-Gibbs algorithm, with a Gibbs
step for the trajectory and a Metropolis--Hastings step for the parameters. 
We implement such strategy in \Cref{subsec:inference_param}.

In this framework, two decompositions are in principle available for the joint conditional density of the signal. The \emph{backward} decomposition reads
\begin{equation}
\begin{aligned}
 p(x_{0:T}|y_{0:T}) 
 = &\, p(x_T|y_{0:T}) \prod_{i = 0}^{T-1}p(x_i|x_{i+1}, y_{0:T}) 
 =p(x_T|y_{0:T}) \prod_{i = 0}^{T-1}p(x_i|x_{i+1}, y_{0:i}) 
\end{aligned}
\end{equation} 
where the backward kernel, in virtue of Bayes' Theorem, can be written 
\begin{equation}\label{backward kernel}
p(x_i|x_{i+1}, y_{0:i}) 
=\frac
{p(x_{i+1}| x_{i})p(x_{i}|y_{0:i})}
{p(x_{i+1}|y_{0:i})}
=\frac{P_{\Delta}(x_{i+1}|x_{i})\nu_{i|0:i}(x_{i})}{\nu_{i+1|0:i}(x_{i+1})}
\end{equation} 
and $\nu_{i|0:i}$ and $\nu_{i+1|0:i}$ are the filtering and prediction densities as in \eqref{prediction in thm}-\eqref{filtering in thm}. 
The typical approach to smoothing using this decomposition is generally known as the \gls{FFBS} algorithm (cf.~\citealp{chopinIntroductionSequentialMonte2020,book:54642}),
which consists in first implementing a forward pass by filtering the whole signal, and then generating a draw of the conditional trajectory by means of a backward pass which starts from a sample from the last filtering distribution and uses \eqref{backward kernel} repeatedly.

An alternative approach uses the \emph{forward} decomposition 
\begin{equation}
\begin{aligned}
 p(x_{0:T}|y_{0:T}) 
 = &\, p(x_0|y_{0:T}) \prod_{i = i}^{T}p(x_i|x_{i-1}, y_{0:T}) 
 =p(x_0|y_{0:T}) \prod_{i = 0}^{T}p(x_i|x_{i-1}, y_{i:T}),
\end{aligned}
\end{equation} 
where the forward kernel, again by Bayes' Theorem, can be written 
\begin{equation}\label{forward kernel}
p(x_i|x_{i-1}, y_{i:T}) 
=\frac
{p(x_{i}| x_{i-1})p(y_{i:T}|x_i)}
{p(y_{i:T}|x_{i-1})}
=\frac{P_{\Delta}(x_{i}|x_{i-1})p(y_{i}|x_{i})p(y_{i+1:T}|x_{i})}{p(y_{i:T}|x_{i-1})}.
\end{equation} 
The cost-to-go functions appearing on the right hand side are available through  \Cref{prop: prediction for smoothing}.
Here the approach to obtaining a sample from the joint smoothing distribution is to compute the cost-to-go functions backward and to perform forward sampling.

In both cases, the finite mixture representations obtained in \Cref{sec: main results} (filtering density for the backward decomposition, cost-to-go functions for the forward decomposition) need to be combined with the transition density of the signal, so further analysis on how to simulate in pratice from $p(x_i|x_{i+1}, y_{0:i})$ 
or $p(x_i|x_{i-1}, y_{i:T})$ 
depends on the specific model. 
We show how this can be done for the \gls{CIR} and \gls{WF} model.


\subsection{Cox--Ingersoll--Ross signals}\label{sec: CIR simulation}

Under the specifications of Section \ref{sec: illustration_CIR}, the unconditional transition density of the signal can be written as the infinite mixture of gamma densities
\begin{equation}\label{eq:uncnd_tr_CIR}
\begin{aligned}
P_{\Delta}(x_{i+1}|x_{i})
=\sum_{k\ge0}\text{Po}(k; \Theta_{\Delta}'x_{i})\text{Ga}(x_{i+1}; \delta/2+k,e^{2\gamma\Delta}\Theta_{\Delta}'), \quad \quad 
\Theta_{\Delta}'=\frac{\gamma/\sigma^{2}}{e^{2 \gamma \Delta}-1},
\end{aligned}
\end{equation} 
which implies, using \eqref{backward kernel}, that
\begin{equation}
\begin{aligned}
p(x_{i}|x_{i+1}, y_{0:i}) 
=&\,
\frac{\sum_{k\ge0}\text{Ga}(x_{i+1}; \delta/2+k,e^{2\gamma\Delta}\Theta_{\Delta}')\text{Po}(k; \Theta_{\Delta}'x_{i})\nu_{i|0:i}(x_{i})}{\nu_{i+1|0:i}(x_{i+1})}.
\end{aligned}
\end{equation} 
When $\nu_{i|0:i}(x_{i})$ is as in \eqref{CIR filtering}, the terms in $x_{i}$ yield
\begin{equation}\nonumber
\begin{aligned}
\text{Po}(k; \Theta_{\Delta}'x_{i})\nu_{i|0:i}(x_{i})
=&\,
\sum_{y_{i}\le m\le N_i}
w_{m}^{(i)} \mu_{m, \vartheta_{i|0:i}}(k)
\Ga(x_{i};\delta/2+m+k,\vartheta_{i|0:i}+\Theta_{\Delta}')
\end{aligned}
\end{equation} 
where we have used Bayes' Theorem to exploit the conjugacy of the Gamma--Poisson model (cf.~Section \ref{sec: illustration_CIR}) and where
\begin{equation}\nonumber
\mu_{m,\vartheta_{i|0:i}}(k)
=\text{NB}\bigg(k; \frac{\delta}{2}+m,\frac{\vartheta_{i|0:i}}{\vartheta_{i|0:i}+\Theta_{\Delta}'}\bigg)
\end{equation} 
is the marginal distribution of $k$ when $k|x_{i}\sim\text{Po}(k;\Theta_{\Delta}'x_{i})$ and $x_{i}\sim \Ga(x_{i};\delta/2+m,\vartheta_{i|0:i})$.
It follows that
\begin{equation}
\label{eq:joint_smoothing_CIR}
p(x_{i}|x_{i+1}, y_{0:i}) 
=\sum_{k\ge0}
\tilde w_{k}(x_{i+1},y_{0:i})
\sum_{y_{i}\le m\le N_i}
\hat w_{m,k}^{(i)}
\Ga(x_{i};\delta/2+m+k,\vartheta_{i|0:i}+\Theta_{\Delta}')
\end{equation} 
where
\begin{align}
\tilde w_{k}(x_{i+1},y_{0:i})
=&\,
\frac{\text{Ga}(x_{i+1}; \delta/2+k,e^{2\gamma\Delta}\Theta_{\Delta}')}{\nu_{i+1|0:i}(x_{i+1})}, \quad \quad 
\hat w_{m,k}^{(i)}
= w_{m}^{(i)} \mu_{m, \vartheta_{i|0:i}}(k),
\end{align}
and $\nu_{i+1|0:i}(x_{i+1})$ is as in \eqref{CIR prediction}.

\Cref{eq:joint_smoothing_CIR} is again an infinite mixture of gamma densities, implying that a simple algorithm for simulating from $p(x_{i}|x_{i+1}, y_{0:i}) $ is to sample the double index $(k, m)$ from the discrete mixture with probability masses equal to the weights in \eqref{eq:joint_smoothing_CIR}, and then sample $x_i$ from a gamma distribution whose parameters are determined by the drawn $(k, m)$. A schematic version of this strategy is as follows: 
\begin{itemize}\label{CIR scheme}
\item draw $U\sim \text{U}(0,1)$
\item set $\kappa, M$ to be the minimum values $j,p$ that yield 
$$ U\le  \sum_{k\le j-1}
\tilde w_{k}(x_{i+1},y_{0:i})\sum_{y_{i}\le m\le N_i}
\hat w_{m,k}^{(i)} + \tilde w_{j}(x_{i+1},y_{0:i})\sum_{y_{i}\le m\le p}
\hat w_{m,j}^{(i)}$$
\item draw $X_{i}\sim \Ga(x_{i};\delta/2+\kappa+M,\theta+\lambda+\Theta_{\Delta}')$.
\end{itemize}

\subsection{Wright--Fisher signals}\label{sec: WF simulation}

A similar strategy could in principle be applied to the \gls{WF} model as well, under the specifications of \Cref{sec: WF illustr}. The unconditional transition density of the \gls{WF} diffusion can be written (cf., e.g., \citealp{ethierTransitionFunctionFlemingViot1993})
\begin{equation}\label{eq: WF transition dens}
P_{\Delta}(\xx_{i}|\xx_{i-1})=\sum_{m=0}^{\infty} q_{m} \sum_{\ll \in \mathbb{N}^{K}:\ l_{1}+\ldots+l_{K}=m} \operatorname{Multinom}\left(\ll | m, \xx_{i-1}\right) \operatorname{Dir}\left(\xx_{i} | \ll+\xx_{i-1}\right).
\end{equation}
However, the weights $q_{m}$, determined in \cite{griffithsLinesDescentDiffusion1980, tavareLineofdescentGenealogicalProcesses1984}, have an infinite expansion, resulting in an intractable doubly infinite series for the transition density. 
Recently, \cite{jenkins2017exact} devised a strategy for sampling exactly from \eqref{eq: WF transition dens}, by means of an algorithm that deals with alternated weights that decrease only after a certain index.
Leveraging on their results, we are able to obtain an algorithm for sampling from the conditional kernel. First we provide the following representation.

\begin{theorem}\label{thm: WF_joint_smoothing}
Let $\xx_{0:T}$ be as in \Cref{sec: WF illustr}, and let $\yy_{0:T}$ come from a multinomial emission density.
Then \eqref{forward kernel} can be written as the mixture of Dirichlet densities
\begin{equation}\label{eq: WF_forward_kernel}
p(\xx_{i}|\xx_{i-1},\yy_{i:T}) = \sum_{m=0}^\infty \tilde{q}_m \sum_{\ll \in \N^K:\  l_1+\ldots+l_K = m} \tilde  w_{\ll}^m \sum_{\kk \in \aMM_{i|i+1:T}}\tilde w_{\ll, \kk}\mathrm{Dir}(\xx_i|\ll + \xx_{i-1}+\yy_i+\kk).
\end{equation}
\end{theorem}
Explicit expressions for the weights $\tilde{q}_m, \tilde  w_{\ll}^m$ and $\tilde w_{\ll, \kk}$ in \eqref{eq: WF_forward_kernel}, together with a proof of the statement, are provided in the Appendix. 
The two inner sums in the above representation are finite mixtures from which sampling is straightforward. 
The final step is to sample from the infinite mixture with weights $\tilde q_m$.
The strategy proposed in \cite{jenkins2017exact} can be adapted to this aim, upon observing that we can bracket the discrete density of $m$ by two convergent series and use the strategy for dealing with alternated series from \cite{Devroye1986a}. 
A detailed explanation and an explicit sampling algorithm are provided in the appendix.
Note that this approach also underlies the foundational work in \cite{beskosRetrospectiveExactSimulation2006} for simulation of diffusion processes.

\section{Implementation and acceleration}\label{sec: application}

\subsection{Computational challenges}

The expressions we have obtained in Section \ref{sec: main results} are similar in spirit to the well-known Baum-Welch algorithms (cf.~Introduction), with the difference that in our case the finite state space is a subset of $\Z_{+}^{K}$ of varying dimension instead of fixed.
In practice, each component in the involved mixture distributions is associated to a node of the grid in $\Z_{+}^{K}$ given by the current set of active indices. 
As already noted by \cite{chaleyat2006computable}, each update operation with new data increases the nominal number of components in the filtering mixture distribution by shifting the mixture to a  higher set of nodes, where the upward shift represents the accumulation of further information, and associating null weights to all nodes below the shifted set of active indices. 
All these nodes then become active indices upon the prediction operation, thus increasing the effective number of mixture components (those with non null weight). 
A similar intuition applies to the evaluation of the cost-to-go functions, which are combined with the filtering distributions to form the marginal smoothing distributions.

To be more specific, for the two models illustrated in Section \ref{sec: illustration}, the number of mixture components in the filtering distributions evolves as $| \MM_n| = \prod_{k=0}^K (1 + \sum_{i=1}^n y_{i,k})$, where $y_{i, k}$ is the $k$-th coordinate of $\yy_{i}$ observed at time $i$ and $K$ is the dimension of the signal space, hence the polynomial complexity mentioned in the Introduction.
On the implementation side, an important aspect to note is that the prediction step is more expensive than the update step, as at each iteration it involves computing the probability that the dual process reaches any point $\mm$ in $\B(\MM_{i|0:i})$  from any point in $\MM_{i|0:i}$ which is not lower than $\mm$ (cf.~\eqref{below Lambda},\eqref{vartheta and M-sets} and \Cref{prop:rec_filtering}).
When new data arrive at a constant frequency, it is possible to limit the cost of the prediction operation by storing the transition terms $p_{\mm, \nn}(\Delta)$, which can be used multiple times during the successive iterations. 
Nonetheless, the rapid growth in their number 
renders this strategy daunting in terms of memory storage.
A close inspection reveals that the $p_{\mm, \nn}(\Delta)$ are themselves a product of a smaller number of terms (cf.~Lemma 
\textcolor{blue}{S1.1} 
in the supplementary material). 
This number grows only quadratically with the sum of all observations and can be saved, thus improving the computational efficiency.

A major point that allows to substantially improve the computation is however given by the fact that each transition of the death process will assign a non negligible probability weight only to a small number of nodes, the intuition being that as the time interval becomes large, the probability mass progressively concentrates on the grid origin $(0,\ldots,0)\in \Z_{+}^{K}$. 
Cf.~Figure 4 in \cite{ascolaniPredictiveInferenceFleming2020}. 
As a result, the number of mixture components in the filtering distributions with non negligible weight may be roughly stationary or increase at a much lower rate, depending on the ratio between the time lag between observations and the number of observations per collection time. 
Cf.~\Cref{fig:saturation} here and Section 4 in \cite{Papaspiliopoulos2014a}. 
A similar behaviour is expected from the cost-to-go functions, which are a linear combination of functions with many negligible coefficients relative to the largest. 
This suggest that a great computational improvement can be obtained by informedly pruning the mixtures of those components having negligible weight, which we develop in \Cref{sec:pruning}.

Another hurdle in the computation, of a more technical flavour, is due to the alternating signs appearing in the death process transition probabilities (cf.~Lemma 
\textcolor{blue}{S1.1} 
in the supplementary material), susceptible to both over and underflow.
Here we solve this challenge to numerical stability using the \texttt{Nemo.jl} library for arbitrary precision computation \citep{nemo}. 


\subsection{Pruning}\label{sec:pruning}

Based on the intuitions and experimental observations illustrated in the previous section, 
we consider three pruning strategies:
\begin{enumerate}
\item \emph{Fixed number}: retain only a fixed number of components after raking them by weight. This aims at controlling the computational budget at each iteration. 
\item \emph{Fixed mass}: retain the minimum number of components needed to reach a fixed fraction of the total mass.
This aims at controlling the total approximation error incurred by the pruning, and in fact is closely linked to the total variation distance between the exact and the pruned mixture. 

\item \emph{Fixed threshold}: retain only the components whose weight is above a fixed threshold. 
This approach has the advantage of not requiring to rank the mixture components before pruning. 
However, it not obvious how to choose a threshold for the cost-to-go functions, whose coefficients do not sum up to one.

\end{enumerate}

The pruning operation should be performed at each time step, followed by a renormalisation of the remaining weights (in the case of the filtering distributions), so the overall approximation error is the result of several previous approximations. 
However, the error accumulation is counterbalanced by the incorporation of fresh information on the data generating distribution at each collection time. 
In fact, the pruning should be performed after the update with the incoming data, which concentrates the mass of the mixture closer to the true data generating mechanism, thus likely on a fewer number of components. 
This has the further advantage of reducing the computational cost of the prediction step, which is more expensive than the update. 
Hence pruning after the update entails the maximal computational gain as it retains the required mass through the minimal number of components and reduces the number of transitions to be computed. 
Note that since the likelihood is obtained by means of the filtering weights (cf.~Section \ref{sec: likelihood}), the above strategy essentially applies to the computation of the likelihood as well.

\begin{algorithm}[t]
\SetAlgoLined

\textbf{Pruning setting}: ON (approximate filtering) / OFF (exact filtering)

\KwIn{$Y_{0:n}$, $t_{0:n}$ and $\nu = h(x, \oo, \theta_0) \in \mathcal{F}$ for some $\theta_0 \in \Th$}

\KwResult{$\vartheta_{i|0:i}$, $\MM_{i|0:i}$ and $W_{0:n}$ with $W_i = \{w_\mm^i, \mm \in \MM_{i|0:i}\}$ }

\SetKwBlock{Begin}{Initialise}{}

\Begin{

Set $\vartheta_{0|0} = T(Y_0, \theta_0)$ with $T$ as in \eqref{vartheta and M-sets}\\

Set $\MM_{0|0} = \{t(Y_0, \oo)\} = \{\mm^*\}$ and $W_0 = \{1\}$ with $t$ as in \Cref{A: conjugacy}\\

Compute $\vartheta_{1|0}$ from $\vartheta_{0|0}$ as in \eqref{vartheta and M-sets}\\

Set $\MM^* = \B(\MM_{0|0})$ and $W^* = \{p_{\mm^*, \nn}(\Delta, \vartheta_{0|0}), \nn \in \MM^*\}$ with $\B$ as in \eqref{below Lambda} and $p_{\mm, \nn}$ as in (\textcolor{blue}{S1})

}

\For{$i$ from $1$ to $n$}{

\SetKwBlock{Begin}{Update}{}

\Begin{

Set $\vartheta_{i|0:i} = T(Y_i, \vartheta_{i|0:i-1})$\\

Set $W_{i} = \{\frac{w_\mm^* \mu_{\mm, \vartheta_{i|0:i-1}}(Y_i)}{\sum_{\nn \in \MM^*}w_\nn^* \mu_{\nn, \vartheta_{i}}(Y_i)}, \mm \in \MM^*\}$ with $\mu_{\mm, \theta}$ defined as in \eqref{marginals}

Set $\MM_{i|0:i} = \{t(Y_i, \mm), \mm \in \MM^*\}$ and update the labels in $W_i$\\
Copy $\vartheta_{i|0:i}$, $\MM_{i|0:i}$ and $W_{i}$ to be reported as the output
}

\uIf{pruning \emph{ON}}{

Prune $\MM_{i|0:i}$ and remove the corresponding weights in $W_i$\\

Normalise the weights in $W_i$

}

\SetKwBlock{Begin}{Predict}{}

\Begin{

Compute $\vartheta_{i+1|0:i}$ from $\vartheta_{i|0:i}$\\

Set $\MM^* = \B(\MM_{i|0:i})$ and $W^* = \left\{\displaystyle{\sum_{\mm \in \MM_{i|0:i}, \mm \ge \nn}}w_{\mm}^ip_{\mm, \nn}(\Delta,\vartheta_{i|0:i}), \nn \in \MM^*\right\}$
}
}
\begin{quote}
\caption{\small Filtering\label{algo_filtering}}
\end{quote}
\end{algorithm}

We lay out two practical algorithms for implementing the recursions of Section \ref{sec: main results}, both with the option of pruning. If kept off, the algorithms thus provide an exact evaluation of the respective distributions. The pseudo code for computing the filtering densities is provided in Algorithm \ref{algo_filtering}, while that for the smoothing densities in Algorithm \ref{algo_smoothing}. 
Algorithm S1 in the Supplementary Material illustrates how to modify Algorithm \ref{algo_filtering} to compute the likelihood.

\begin{algorithm}[h!]
\SetAlgoLined

\textbf{Pruning setting}: ON (approximate smoothing) / OFF (exact smoothing)

\KwIn{$Y_{0:n}$, $t_{0:n}$ and $\nu = h(x, \oo, \theta_0) \in \mathcal{F}$ for some $\theta_0 \in \Th$}

\KwResult{$\vartheta_{i|0:n}$, $\MM_{i|0:n}$ and $W_{0:n|0:n}$ with $W_{i|0:n} = \{w_\mm^i, \mm \in \MM_{i|0:n}\}$ }

\SetKwBlock{Begin}{Initialise}{}

\Begin{

Obtain $\vartheta_{i|0:i}$, $\MM_{i|0:i}$ and $W_{0:n}$ with $W_i = \{w_\mm^i, \mm \in \MM_{i|0:i}\}$ from the filtering algorithm, with pruning if relevant.

Set $\overleftarrow\vartheta_{n|n} = T(Y_n, \theta_0)$ with $T$ as in \eqref{vartheta and M-sets}\\

Set $\overleftarrow\MM_{n|n} = \{t(Y_n, \oo)\} = \{\mm^*\}$ and $\overleftarrow W_n = \{1\}$ with $t$ as in \Cref{A: conjugacy}\\

Compute $\overleftarrow\vartheta_{n-1|n}$ from $\overleftarrow\vartheta_{n|n}$ as in \eqref{vartheta and M-sets}\\

Set $\overleftarrow\MM^* = \B(\overleftarrow\MM_{n|n})$ and $\overleftarrow W^* = \{p_{\mm^*, \nn}(\Delta, \overleftarrow\vartheta_{n|n}), \nn \in \overleftarrow\MM^*\}$ with $\B$ as in \eqref{below Lambda} and $p_{\mm, \nn}$ as in (\textcolor{blue}{S1}) 
}

\For{$i$ from $n-1$ to $1$}{

\SetKwBlock{Begin}{Update}{}

\Begin{

Set $\overleftarrow\vartheta_{i|i:n} = T(Y_i, \vartheta_{i|i+1:n})$\\

Set $\overleftarrow W_{i} = \{w_\mm^* \mu_{\mm, \overleftarrow\vartheta_{i|i+1:T}(Y_i)}, \mm \in \overleftarrow\MM^*\}$ with $\mu_{\mm, \theta}$ defined as in \eqref{marginals}

Set $\overleftarrow\MM_{i|i:n} = \{t(Y_i, \mm), \mm \in \MM^*\}$\\

}

\uIf{pruning \emph{ON}}{

Prune $\overleftarrow\MM_{i|i:n}$ and remove the corresponding values in $\overleftarrow W_i$\\

}

\SetKwBlock{Begin}{Predict}{}

\Begin{

Compute $\overleftarrow \vartheta_{i-1|i:n}$ from $\overleftarrow\vartheta_{i|i:n}$\\

Set $\overleftarrow \MM^* = \B(\overleftarrow\MM_{i|i:n})$ and $\overleftarrow W^* = \left\{\displaystyle{\sum_{\mm \in \overleftarrow\MM_{i|i:n}, \mm \ge \nn}}w_{\mm}^ip_{\mm, \nn}(\Delta,\overleftarrow\vartheta_{i|i:n}), \nn \in \overleftarrow\MM^*\right\}$
}
}
\For{$i$ from $1$ to $n-1$}{
Set $\vartheta_{i|0:n} = e(\overleftarrow\vartheta_{i|i+1:n}, \vartheta_{i|0:i})$

Set $\MM_{i|0:n} = \{d(\mm, \nn); \mm \in \overleftarrow\MM_{i|i+1:n}, \nn \in \MM_{i|0:i}\}$ considering only the unique values.

Set $\displaystyle{W_{i|0:n} = \{w_{\ll}^i \propto \sum_{\mm \in \overleftarrow\MM_{i|i+1:n}, \nn \in \MM_{i|0:i}, d(\mm, \nn) = \ll}\overleftarrow w_{\mm}^{(i+1)}w_{\nn}^{(i)}C_{\mm, \nn, \overleftarrow\vartheta_{i|i+1:n}, \vartheta_{i|0:i}}, \ll \in \MM_{i|0:n}\}}$
}
\begin{quote}
\caption{\small Smoothing\label{algo_smoothing}}
\end{quote}
\end{algorithm}



\section{Numerical experiments}\label{sec: numerical experiments}


\subsection{Retrieving the signal}\label{subsec:inference}

We illustrate inference on the signal trajectory for the models illustrated in \Cref{sec: illustration}. 
We reparametrise the \gls{CIR} model by considering
\begin{equation}\label{CIR reparameterised}
 \d X_t = a(b-X_t)\, \d t + s\sqrt{X_t}\, \d B_t
\end{equation}
in place of \eqref{CIR SDE}, 
where $b$ is the mean, $a$ the speed of adjustment towards the mean and $s$ controls the volatility, which respectively correspond to $   a = 2\gamma, b = \delta\sigma^2/(2\gamma)$ and $s = 2\sigma$.
We simulate the trajectory of a \gls{CIR} process starting from $X=3$ with $a = 5$, $b = 9.6$ and $s = 8$, corresponding to a  $\text{Gamma}(1.5,0.15625)$ stationary distribution. 
We draw 10 observations at each of 200 collection times separated by 0.011 seconds.

For the \gls{WF}, we implement\footnote{Our code is available at \url{https://github.com/konkam/ExactWrightFisher.jl}.} the exact simulation scheme of \cite{jenkins2017exact} to simulate a signal with $K=3$ components, initialising the process at random from a Dirichlet($0.3, 0.3, 0.3$) distribution, and draw 15 observations at each of 10 observation times separated by 0.1 seconds.

\Cref{algo_works} shows the \emph{exact} filtering for both models, obtained following the parameter updates detailed in \Cref{sec: illustration} by means of Algorithm \ref{algo_filtering} (with pruning off). The non observed trajectories of the signals are correctly recovered by the filter.

\begin{figure}[t!]
\begin{center}
\includegraphics[width = .7\textwidth]{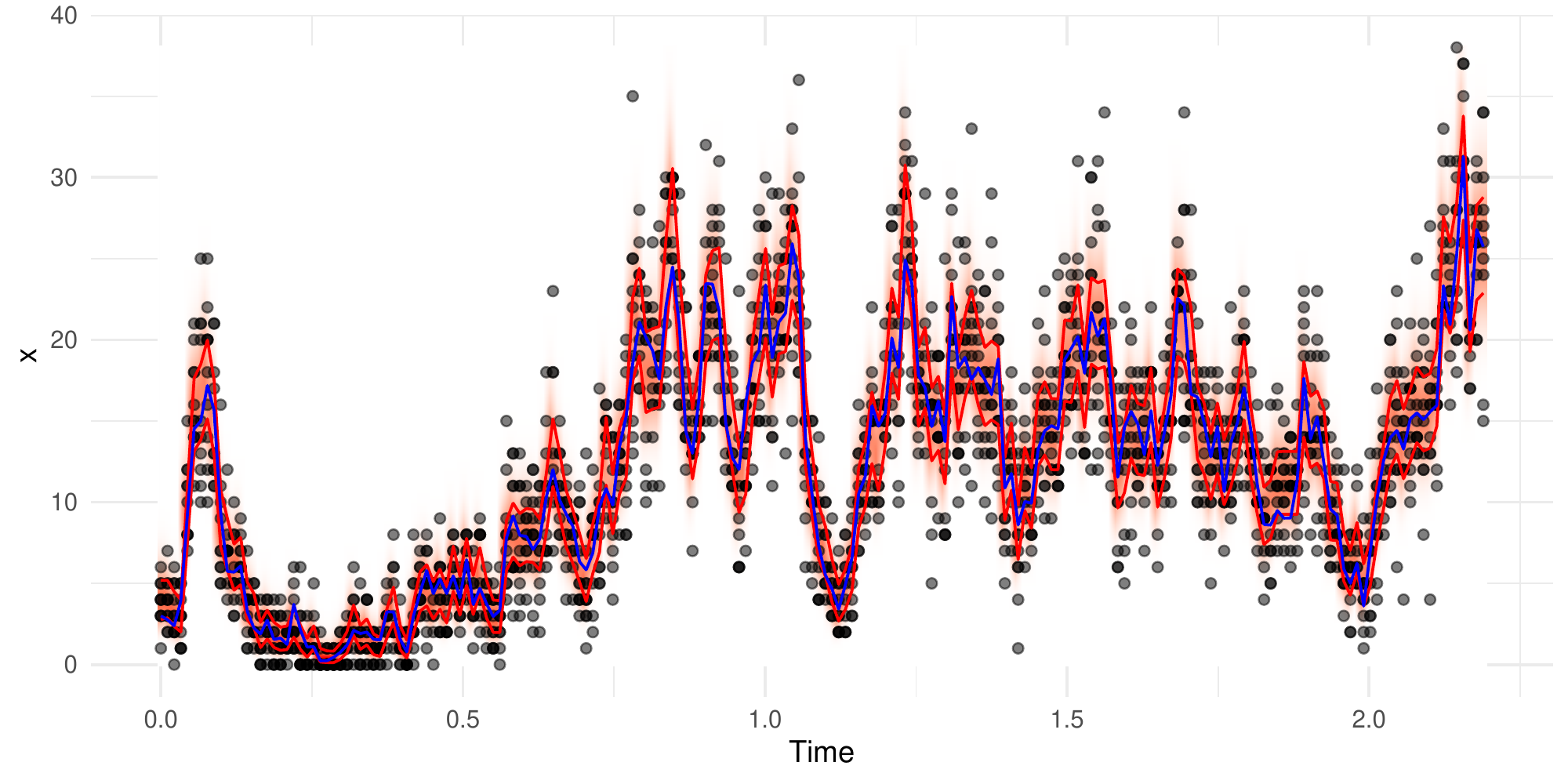}\\
\includegraphics[width = .7\textwidth]{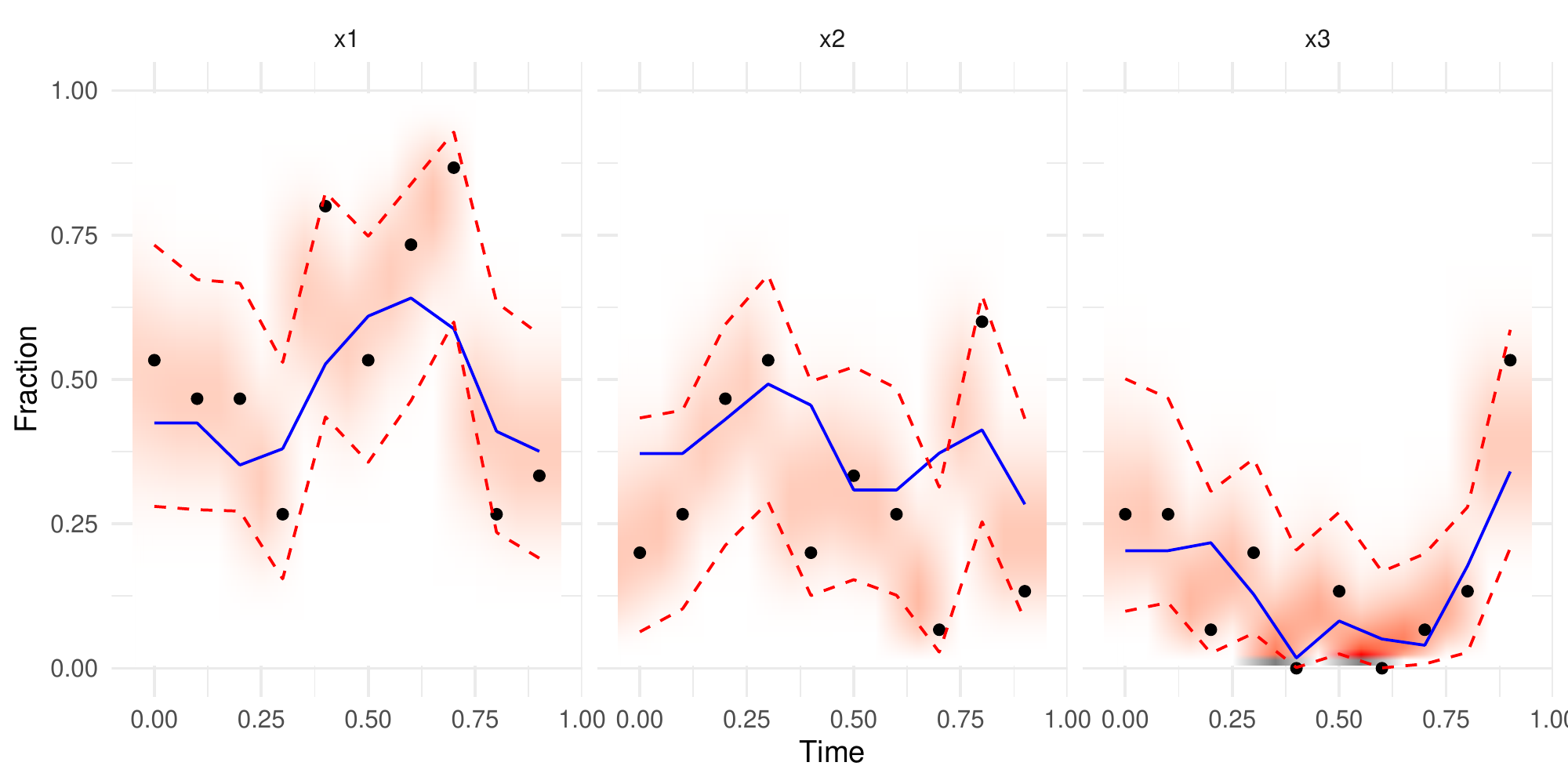}
\begin{quote}
\caption{ Hidden signal (blue solid), data points (bullets) and 95\% pointwise credible intervals (red dashed) derived from the filtering distribution for the \gls{CIR} (top) and \gls{WF} (bottom) simulated datasets.
A heat map also represents the filtering densities, darker red indicating higher density.
For the \gls{WF}, each panel plots the marginal for one coordinate, and the single bullet is the proportion of the observations from the corresponding type.
\label{algo_works}}
\end{quote}
\end{center}
\end{figure}

\Cref{joint_smoothing_algo_works} shows a few trajectories sampled from the joint smoothing density of the CIR model, i.e., the conditional distribution of the signal given the data. This is performed by following the strategy outlined in \Cref{sec:joint_inference} and exploits the correlation with neighbouring time points, thus improving the estimates obtained by means of the marginal smoothing distributions only (\Cref{marginal_smoothing}).

\begin{figure}[t!]
\begin{center}
\includegraphics[width=0.8\textwidth]{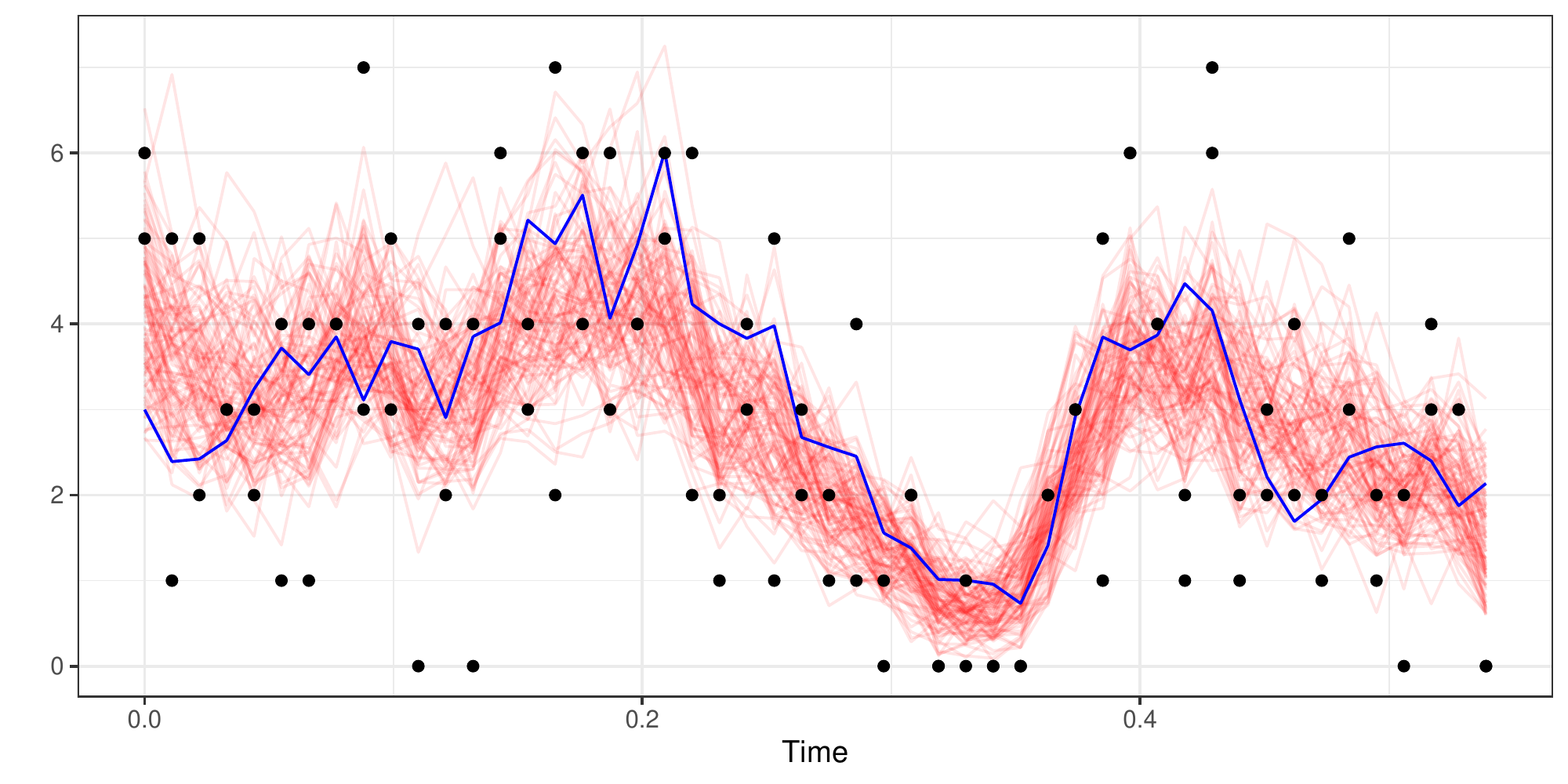}
\begin{quote}
 \caption{ True trajectory (blue) and a few trajectories  sampled from the joint smoothing distribution (red) of the \gls{CIR} model, together with the data (dots). 
 \label{joint_smoothing_algo_works}
 }
\end{quote}
\end{center}
\end{figure}


\subsection{Performance}

We assess the performance of the pruning strategies proposed in \Cref{sec:pruning} applied to the two above models both from a visual and a quantitative perspective. 
\Cref{plot_approximate_likelihood} illustrates the effect of approximating the conditional likelihood, based on the same simulated scenario of the previous section. 
We test the most drastic approximation which retains a single component. 
This degrades the smoothness of the likelihood and perturbs the concavity of the curves, which in turn renders the optimisation problem challenging.
A less extreme approximation, which retains the 10 largest components in the filtering distribution, provide seemingly concave curves whose optima are very close to the exact maximum likelihood estimates.

\begin{figure}[t!]
\begin{center}
 \includegraphics[width=.8\textwidth]{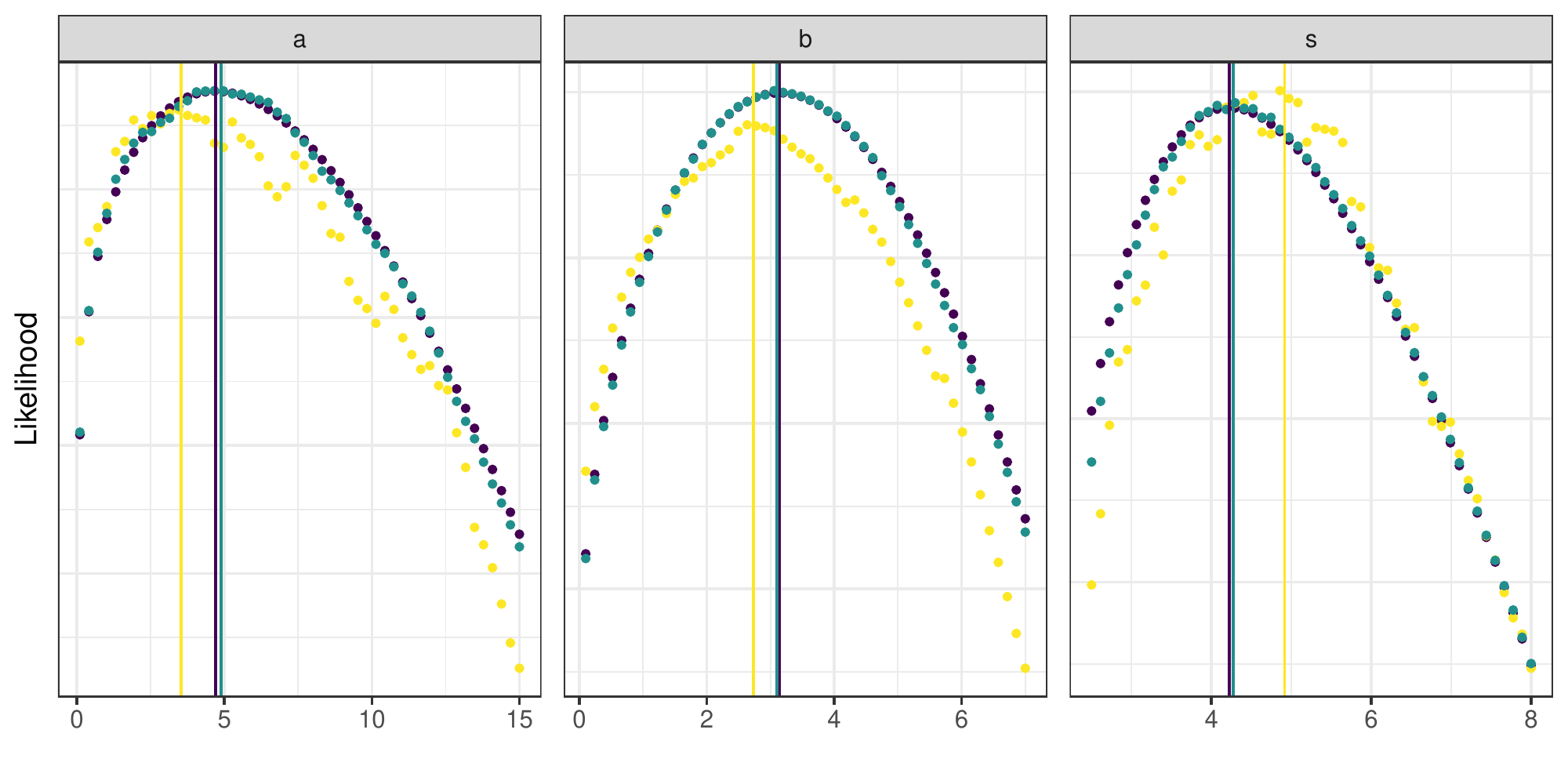}
 \begin{quote}
\caption{ Comparison between exact and approximate likelihood  obtained via pruning, for parameters $a,b,s$ in \eqref{CIR reparameterised}.
The dots represent the conditional likelihood for one parameter, computed on a grid and fixing the other parameters at their true value. 
The vertical lines represent the maximum likelihood estimate obtained by maximising these conditional likelihoods (paired colours). 
The darkest colour represents the exact likelihood, the intermediate colour represent an approximate likelihood obtained by pruning all but the 10 largest components at each time, and the lightest colour represents the likelihood obtained by pruning all but the single largest component at each time.
 The simulated data is the same as used in \Cref{subsec:inference}. \label{plot_approximate_likelihood}}
\end{quote}
\end{center}
\end{figure}

Next we evaluate more systematically the performance of the pruning strategies for filtering, smoothing and likelihood computation, by comparing them against a bootstrap particle filter with adaptive resampling for the likelihood, 
and a \glsfirst{FFBS} particle smoother with adaptive resampling. 
The bootstrap particle filter requires sampling from the transition density of the processes, which is straightforward for the \gls{CIR} process and feasible for the \gls{WF} through the strategy in \cite{jenkins2017exact}.
The \gls{FFBS} particle smoother (or the two-filter particle smoother) requires the evaluation of the transition density of the process.
This is tractable for the \gls{CIR} model, using the representation of the transition density via Bessel functions (see~\eqref{CIR transition with Bessel} below)
but not for the \gls{WF} model, where no straightforward representation appears to be available. Thus, we do not implement the particle smoother for the \gls{WF} model.

We first investigate the performance of the pruning approximations for computing the likelihood. 
As computing the filtering and the likelihood is very similar, the performance of the approximations is almost the same.
To quantify the loss of precision due to the approximation, we compute the absolute error on the likelihood resulting from the pruning, and a root mean squared error for the bootstrap particle filter based on 50 replicates.
To quantify the gain in efficiency with the approximation, we measure the time needed to compute the approximate likelihood of the whole dataset relative to the time needed to compute the exact likelihood.
For the particle filter, we measure the time needed to obtain a single estimate.
\Cref{likelihood_accuracy_vs_relative_computing_time} shows the absolute error on the likelihood plotted against the time needed to compute the approximate likelihood, and provides clear evidence that the pruning strategies entail a gain in terms of computing time by three orders of magnitude for the \gls{WF} case and by four to five orders of magnitude for the \gls{CIR} case. 
At the same time, the pruning strategies offer a drastically better precision than the bootstrap particle filters. 
For the \gls{CIR}, the particle filter runtime necessary to obtain a  precision similar to the pruning strategies seems to be three orders of magnitude larger, while it seems even larger for the \gls{WF}.

To quantify the loss of precision incurred for the smoothing, the simplest approach which extends to a multivariate setting is to compute the $L_2$ distance between the exact and the approximate filtering or smoothing distributions obtained via pruning. 
The $L_2$ distance between two mixtures of gamma or Dirichlet distributions respectively has an explicit analytical expression under certain conditions (see supplementary material Section \textcolor{blue}{S2}
, which are satisfied in our setting. 
As there is one smoothing distribution per observation time, we consider the maximum over time of the above distances. 
In order to compare these to the exact and approximate smoothing distributions obtained with pruning, we exploit the fact that a sample from the particle smoother comes from a mixture of gamma distributions.
We estimate the smoothing distributions associated with the particle smoother by means of gamma kernel density estimates \citep{chen2000probability, Boone2014}.
The particle smoothing distribution is then estimated by a mixture of gamma components with equal weights, and it is possible to compute the $L_2$ distance with the exact smoothing distribution as with the pruning approximations.
\Cref{smoothing_performance} shows that the pruning approximations to the smoothing distributions are even more efficient than for the likelihood and the filtering.
This is because the particle smoothing computing time seems dominated by the evaluation of the Bessel function in the \gls{CIR} transition density, while our exact dual smoothing algorithm avoids evaluating the transition density altogether.

\begin{figure}[t]
\begin{center}
\includegraphics[width = .8\textwidth]{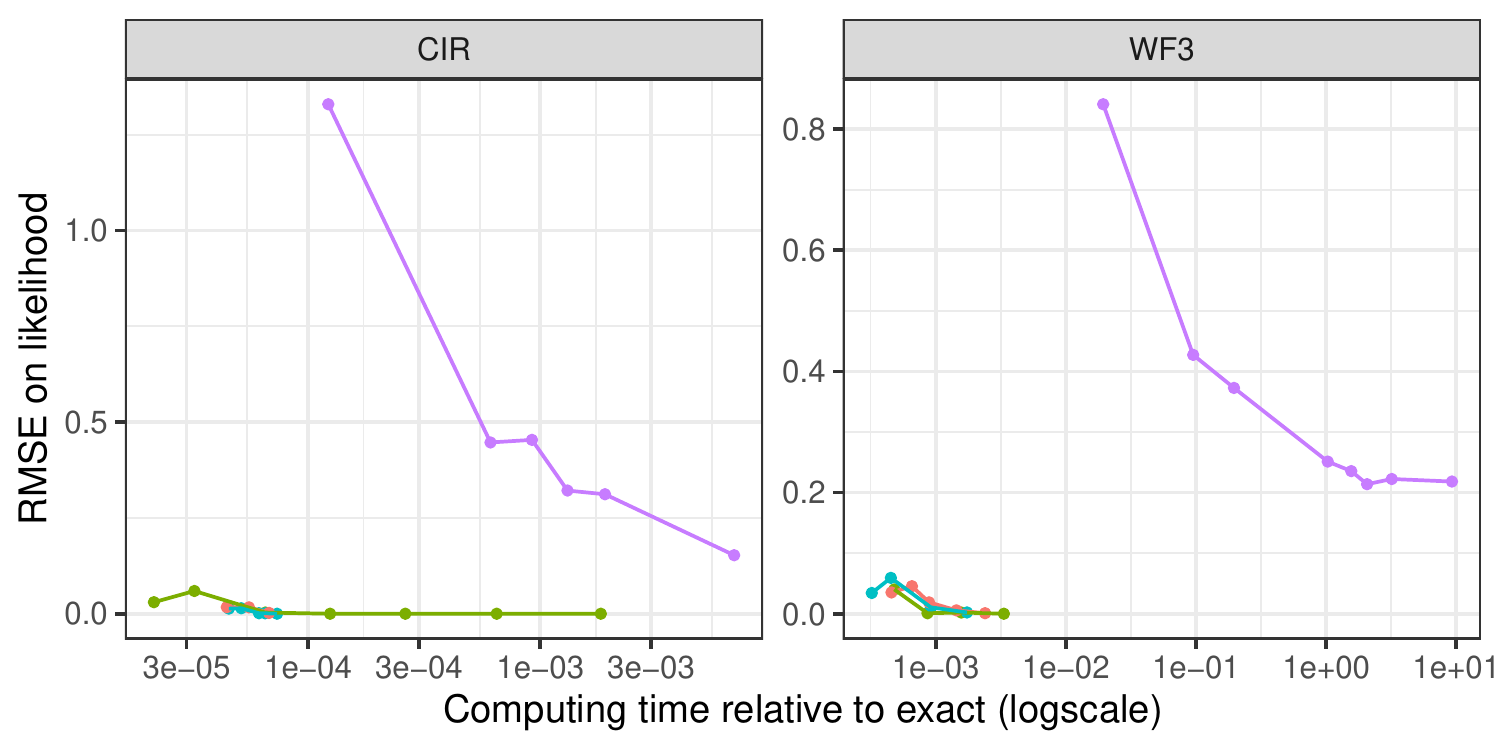}

\begin{quote}
\caption{ Accuracy of the likelihood estimate (RMSE) against computing time (relative to computing the exact likelihood), for the \gls{CIR} (left) and the \gls{WF} (right).
Blue: \emph{fixed threshold} approximations with thresholds  0.01, 0.005, 0.001, 0.0005, 0.0001 for \gls{CIR}  and 0.01, 0.005, 0.001, 0.0001 for \gls{WF}.
Green: \emph{fixed number} approximations with 5, 10, 25, 50, 100, 200, 400 components for \gls{CIR} and
 50, 100, 200, 400 for \gls{WF}.
Red: \emph{fixed fraction} approximation with probability masses 0.95, 0.99, 0.999 for \gls{CIR} and 0.8, 0.9, 0.95, 0.99, 0.999 for \gls{WF}.
Violet: particle filter approximations with 1000, 5000, 7500, 10 000, 15 000, 50 000 particles for \gls{CIR} and 100, 500, 1000, 5000, 7500, 10 000, 15 000, 50 000 for \gls{WF}.
\label{likelihood_accuracy_vs_relative_computing_time}}
\end{quote}
\end{center}
\end{figure}

\begin{figure}[t]
\begin{center}
\includegraphics[width = .45\textwidth]{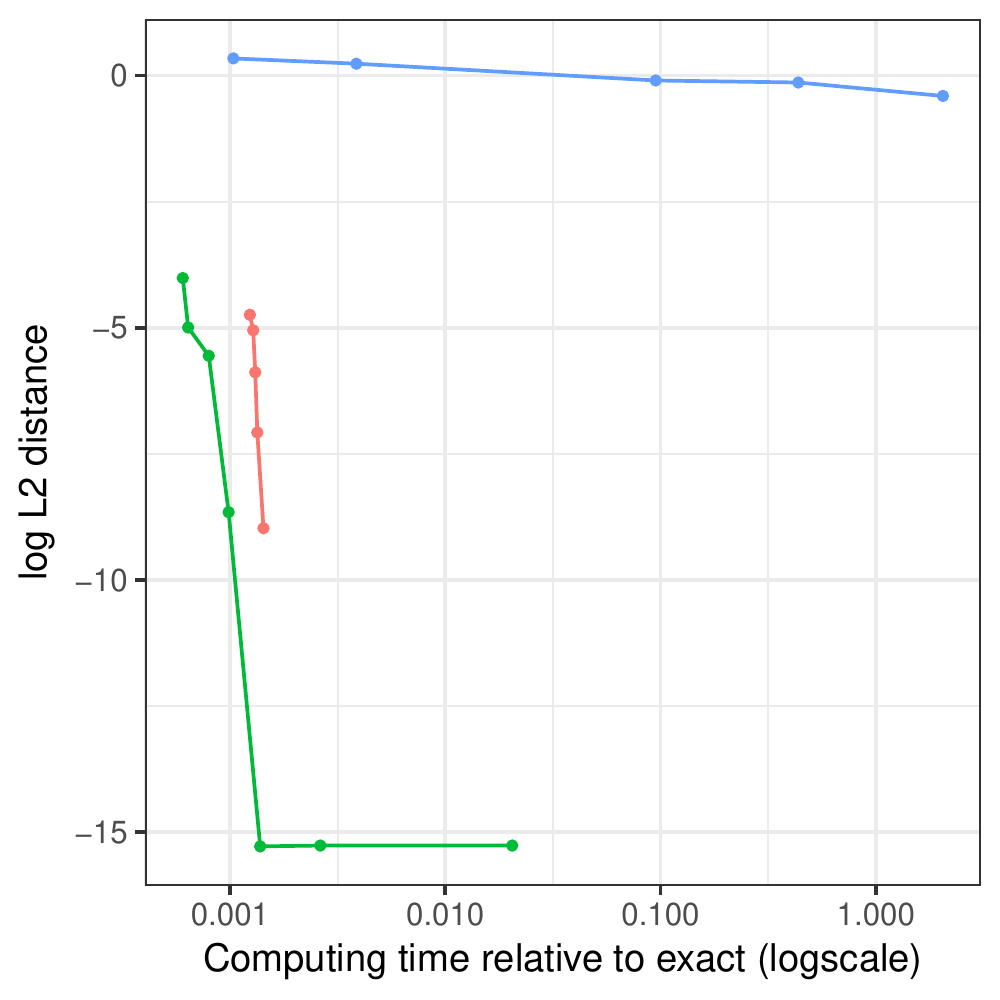}

\begin{quote}
\caption{ Log of the maximum $L_2$ distance between the exact smoothing distributions and the approximated smoothing distributions for the \gls{CIR} model. 
Green: \emph{fixed number} approximation with 5, 10, 25, 50, 100, 200 components. Red: \emph{fixed fraction} approximation with 0.8, 0.9, 0.95, 0.99, 0.999 percent of the total mass retained.
Blue: particle filter approximation with  50, 100, 500, 1000, 2500 particles.
\label{smoothing_performance}}
\end{quote}
\end{center}
\end{figure}

In summary, all pruning strategies are shown to provide very fast and accurate alternatives to the exact strategies, outperforming particle filtering by orders of magnitude.

\subsection{Inference on model parameters}\label{subsec:inference_param}

A fast approximation of the likelihood using the strategy outlined in Section \ref{sec:pruning} unlocks the door to performing inference on the model parameters (denoted \(\psi\) in the Introduction). 
Additionally, using results from \Cref{sec:joint_inference}, one can also address joint inference on signal trajectory and model parameters. 
In this section, we illustrate both inference problems on simulated datasets.
We use the two models considered in previous sections with the specifications of Section \ref{subsec:inference}.

For the \gls{CIR} model, we set parameters $a = 5.0$, $b = 2.4$ and $s = 4.0$ with 2 observations at each of 200 times separated by 0.011 time units.
We use weakly informative exponential priors with rate parameter equal to 0.01 for  each parameter.
For the \gls{WF} model, we choose $\alpha = (1.1, 2.5, 2.1)$, with 15 individuals observed from the population every 0.1 time units over 100 collection times.
We use a weakly informative prior given by a half-normal distribution of location 5 and scale 4 for each parameter.

\subsubsection{Marginal inference on the parameters}

We build an algorithm which uses the filter to estimate the joint posterior on the parameters, via Metropolis--Hastings.
Specifically, we implement a standard symmetric random walk Metropolis--Hastings with Gaussian jumps \gls{MCMC} algorithm to sample from the posterior distribution on the parameters using pruning approximations to the likelihood.
A good jump size was estimated on a pilot run, based on an estimation of the variance-covariance of the posterior near the mode and scaling the sizes to achieve a good acceptance rate, in the range [0.2,0.4].
We ran three chains until convergence, which was estimated from the potential scale reduction factor (implementation from \texttt{MCMCDiagnostics.jl}).
\Cref{fig:ACF_plots} presents the autocorrelation plots for both models.
The approximation of the likelihood is obtained by pruning all but the 10 largest components of the filtering distributions at each time step.
\Cref{MCMC_params} shows that we correctly recover the original parameters. 

\begin{figure}
 \begin{center}
\includegraphics[width = 0.6\textwidth]{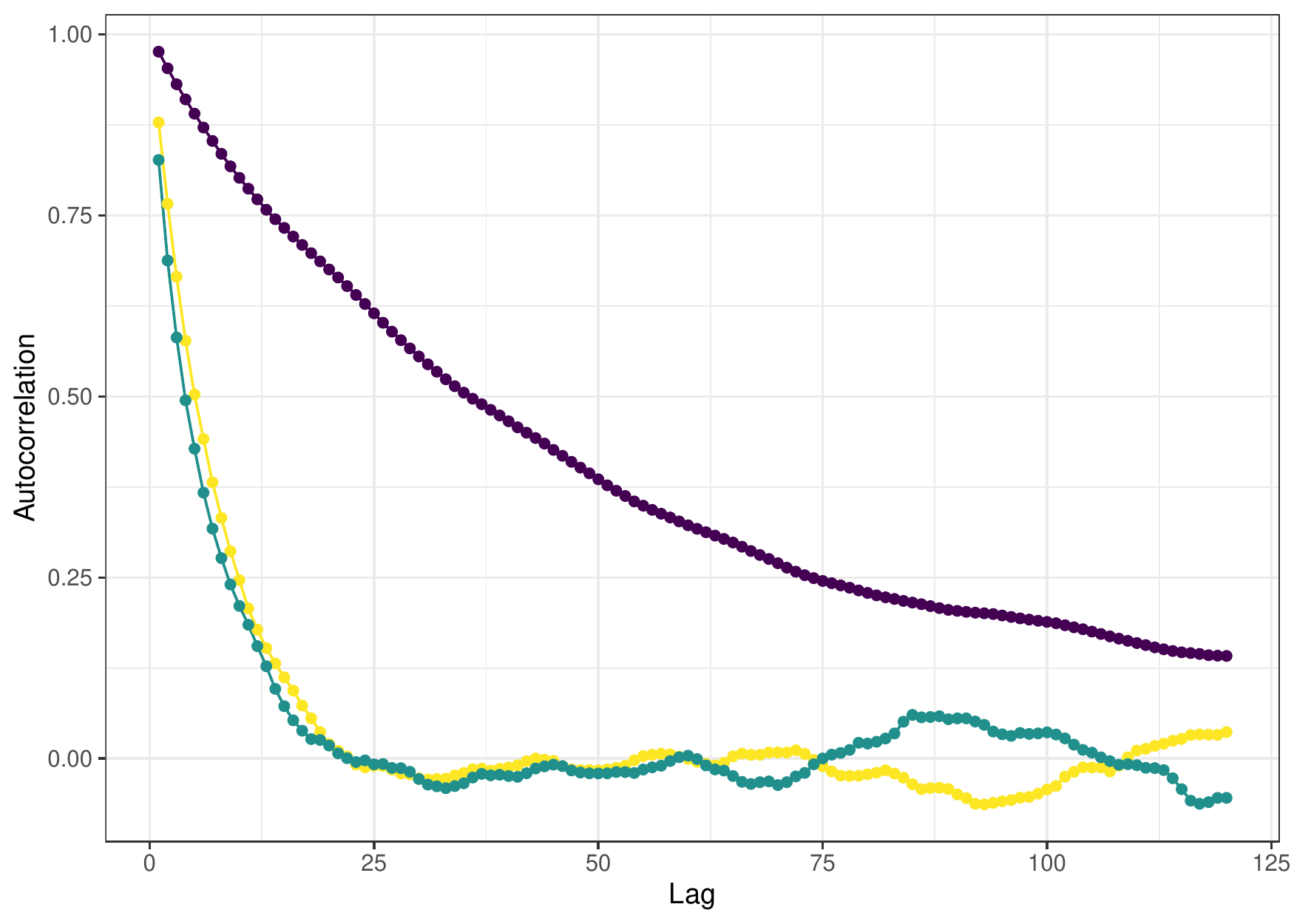}
\begin{quote}
 \caption{Autocorrelation function of the \gls{MCMC} chains for $\theta$ (\gls{CIR}) and for $\alpha_1$ (\gls{WF}). The darkest line corresponds to the joint \gls{CIR} inference, the intermediate line corresponds to the marginal \gls{CIR} inference and the lightest line to the marginal \gls{WF} inference. \label{fig:ACF_plots}}
\end{quote}
 \end{center}
\end{figure}

\begin{figure}
 \begin{center}
  \includegraphics[width=0.4\textwidth]{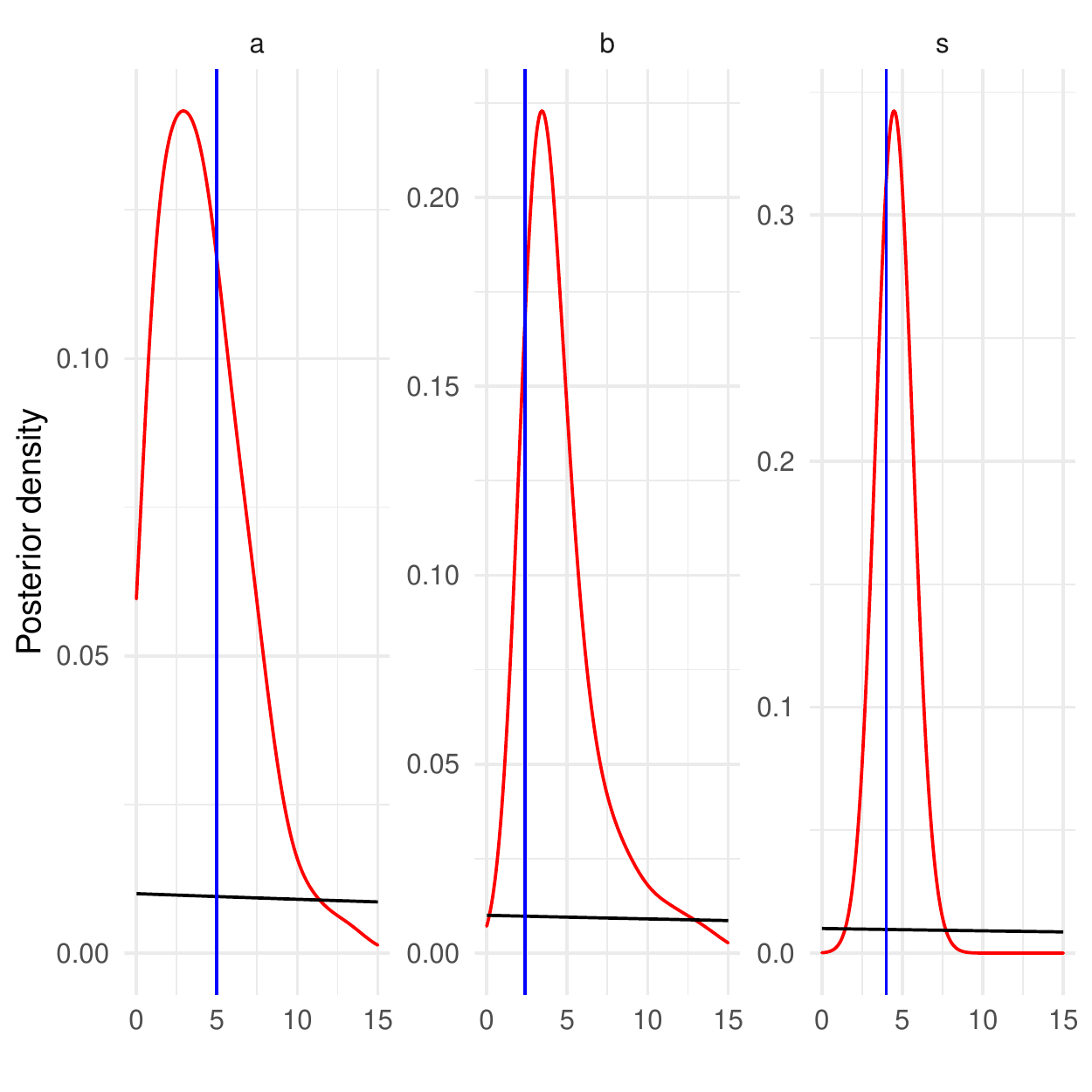}
\includegraphics[width=0.4\textwidth]{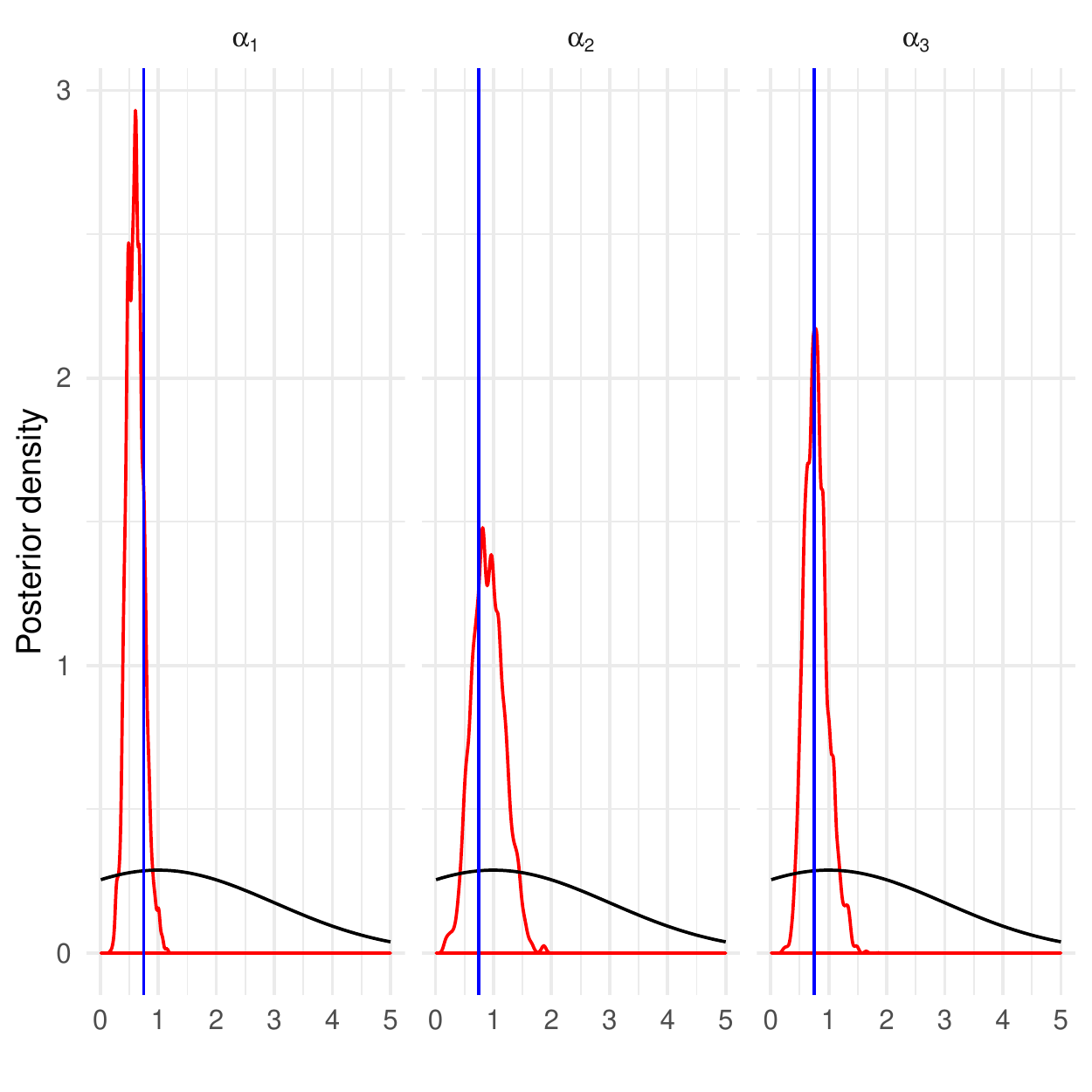}
\begin{quote}
\caption{ Prior (black) and posterior distributions (red) for \gls{CIR} (left) and \gls{WF} (right) parameters, together with true parameters (blue vertical). \label{MCMC_params}}
\end{quote}
\end{center}
\end{figure} 

\subsubsection{Joint inference on signal and parameters}

Using results from \Cref{sec:joint_inference}, we are also able to perform joint inference on the trajectory and the model parameters. 
We illustrate this in the case of the \gls{CIR} model, using a Gibbs sampler with a Metropolis--Hastings step. 
With respect to the previous algorithm for inference on the parameters, this can be viewed as a data augmentation, where states and parameters are sampled jointly.

An outline of the strategy is as follows:
\begin{itemize}
 \item Gibbs
 step: update the signal by sampling a new trajectory from the joint
 smoothing density, conditionally on both parameters and data;
\item Metropolis--Hastings step: update the parameters by computing the likelihood of the data and the trajectory given the parameters.
\end{itemize}

Note that for the Metropolis--Hastings step, we can exploit the following well-known representation for evaluating the transition density of the \gls{CIR} process, written
\begin{equation}\label{CIR transition with Bessel}
 P_{\Delta}(X_{i+1}|X_i)=c\,e^{-u-v} \left (\frac{v}{u}\right)^{q/2} I_{q}(2\sqrt{uv}),
  \label{eq:trCIRbessel}
\end{equation}
where $c=[(1-e^{-2\gamma\Delta})\sigma^2]^{-1}\gamma$,  $q = \delta/2-1$, $u = c X_i e^{-2\gamma\Delta}$, $v = c X_{i+1}$, and $I_{q}(2\sqrt{uv})$ is a modified Bessel function of the first kind of order $q$. 
The absence of such an expression for the transition density of the \gls{WF} process prevents from applying the same strategy directly.

The Metropolis--Hastings step is relatively cheap and consists in evaluating the likelihood for the parameters conditional on the data and a sampled trajectory, which can be done exactly.
The Gibbs step, which consists in simulating a full trajectory conditional on the sampled parameters, can be computationally intensive as it involves summing the weights in the filtering mixtures several times, but can be accelerated using a pruning strategy for the filtering distributions.
We use the same priors as for the inference on the model parameters, and ran the Gibbs algorithm for 10000 iterations. 
We tune the proposal distributions for the Metropolis--Hastings step using the posterior from the marginal sampler for convenience, but this could also be done using a pilot run. 
\Cref{fig:joint_inference_CIR} shows that both the parameter values and the hidden trajectory of the  \gls{CIR} are correctly recovered.
\Cref{fig:ACF_plots} presents the autocorrelation of the chain and shows that the chain used for joint inference had more autocorrelation than those for marginal inference for this simulated dataset. 
The autocorrelation plots for the other parameters are provided in the supplementary material (Figure \textcolor{blue}{S1}) 
.

\begin{figure}
\begin{center}
 \includegraphics[width = 0.47\textwidth]{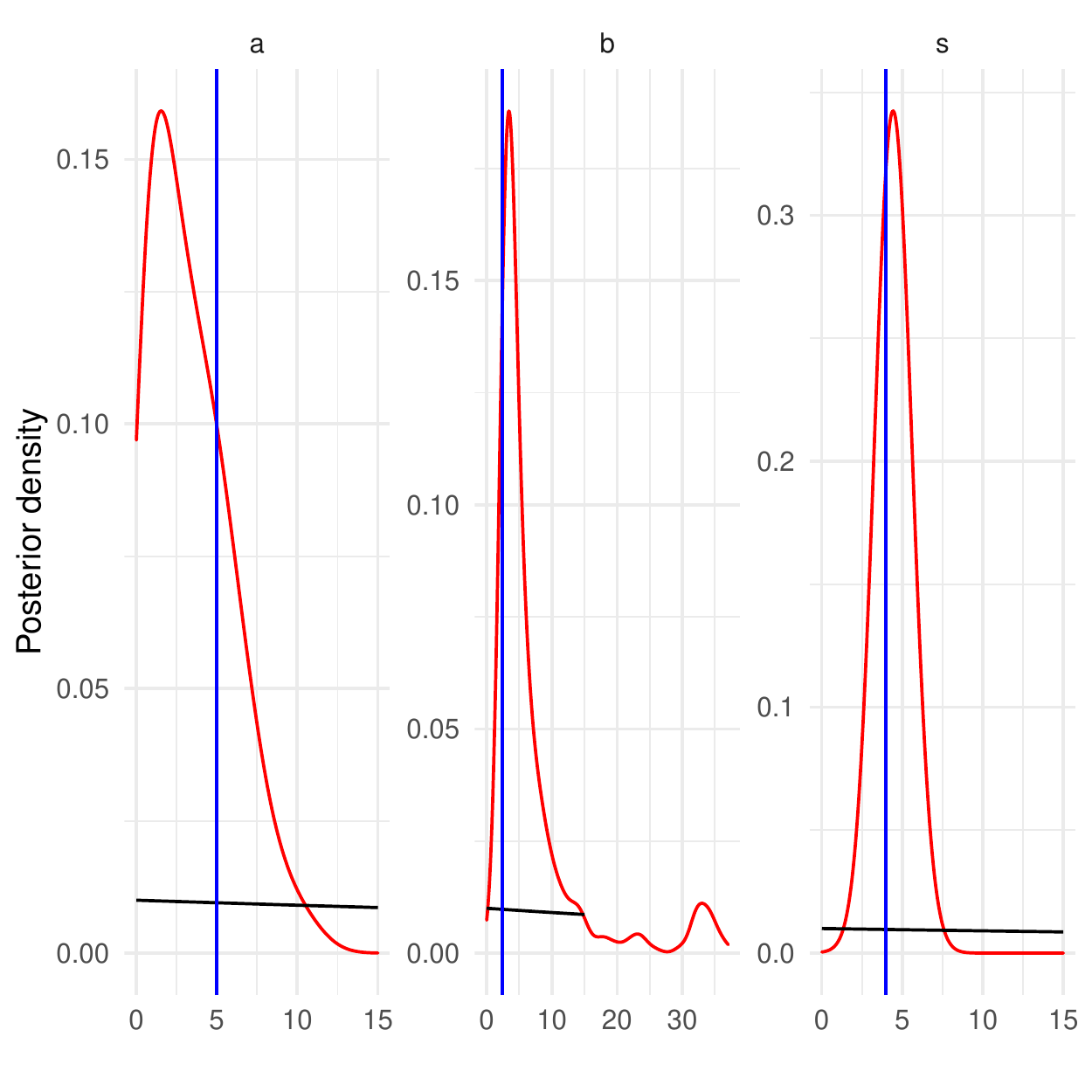}
 \includegraphics[width = 0.47\textwidth]{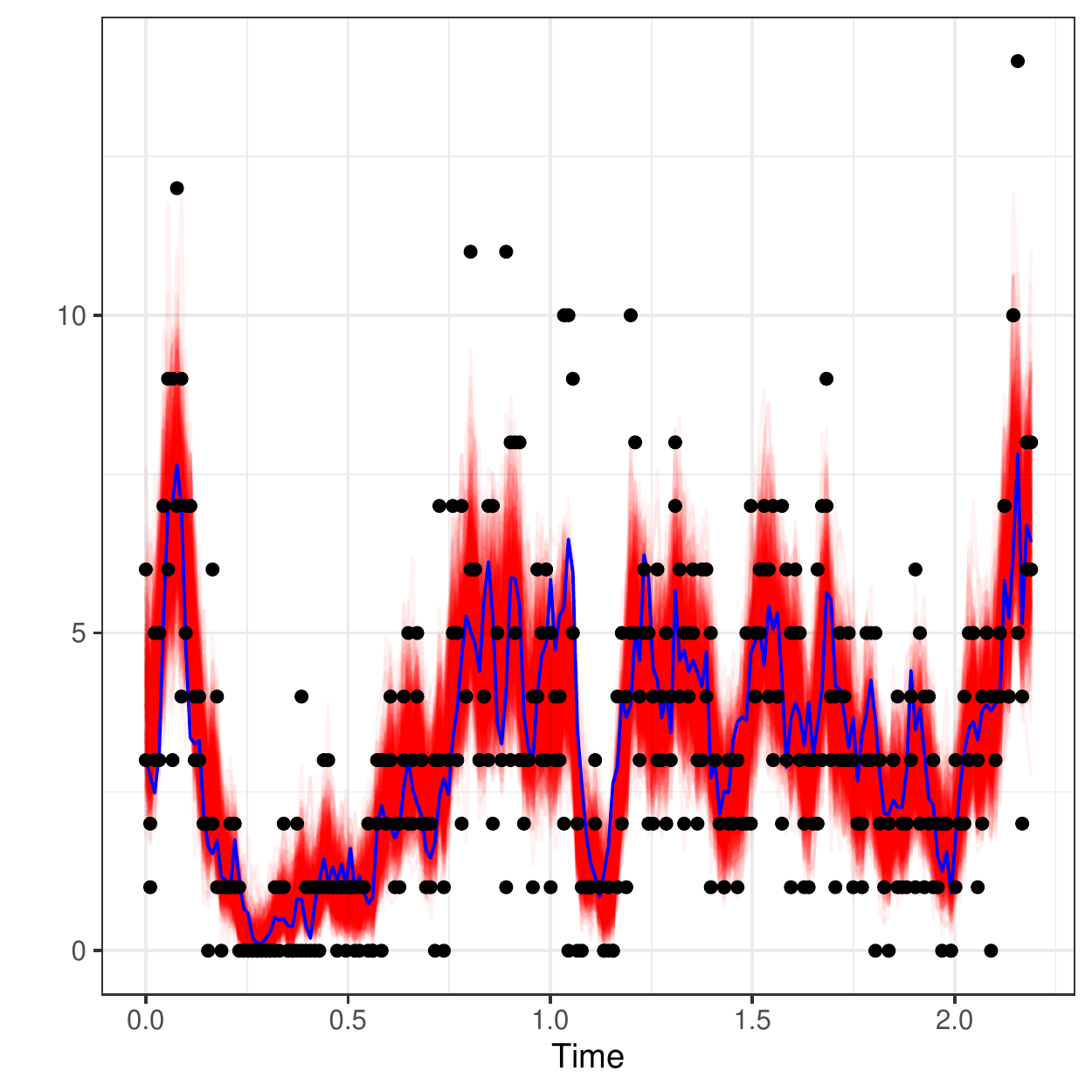}
 \begin{quote}
\caption{\scriptsize Left: prior (black) and posterior distributions (red) for the \gls{CIR} model parameters, along with the true values (blue verticals). 
 Right: a few samples from the posterior distribution (red) for the trajectory and true trajectory (blue) with observed data (black dots).
 \label{fig:joint_inference_CIR}}
\end{quote}
\end{center}
\end{figure}

\section{Discussion}

\cite{Papaspiliopoulos2014a} unveiled that the existence of a certain type of \emph{dual process} for the hidden signal lead to exact expressions of the filtering distributions, expressed in \emph{computable form} as mixtures whose number of components varies with the number of observations but is always finite, thus entailing a polynomial computational complexity. 
Here we have shown that, under essentially the same set of conditions, one can also obtain in computable form expressions for the likelihood and for the marginal smoothing distributions, which can thus be evaluated exactly. 
These expressions allow to design algorithms for addressing inference on the signal trajectory and on the model parameters, detailed in the paper, thus covering the whole agenda of inference for a hidden Markov model. 
We have outlined acceleration strategies based on informed pruning of the mixture representations of the distributions of interest, and showed that the resulting approach vastly outperformed both in accuracy and computing time a generic approximation scheme like particle filtering and smoothing, bringing the computational burden close to linear in the number of observations.
Our algorithms have been made publicly  available via the dedicated Julia package \texttt{DualOptimalFiltering}, which implements a general interface applicable to any 
\gls{HMM} satisfying the conditions outlined in this paper.

Possible direct developments of this work could concern using the exact expressions we obtained to devise efficient particle filters. 
One approach would be to use the pruning approximations to guide a particle filter. 
Another could be to filter the dual process instead of the signal, since the former lives in a discrete state space and would be reasonably easy to sample.


\section*{Acknowledgements}

The third author is partially supported by the Italian Ministry of Education, University and Research (MIUR) through PRIN 2015SNS29B and through ``Dipartimenti di Eccellenza'' grant 2018-2022.

\appendix
\section*{Appendix}

\subsection*{Proof of Theorem \ref{thm: likelihood}}

Using \eqref{prediction in thm} we have
\begin{equation}
\begin{aligned}
\int_{\X}f_{x_{i}}(y_{i})\nu_{i|0:i-1}(x_{i})
=&\int_{\X}f_{x_{i}}(y_{i})\sum_{\mm \in \MM_{i|0:i-1}}w_\mm^{(i)'}g(x, \mm, \vartheta_{i|0:i-1})\\
=&\,
\sum_{\mm \in \MM_{i|0:i-1}}w_\mm^{(i)'}\int_{\X}f_{x_{i}}(y_{i})g(x, \mm, \vartheta_{i|0:i-1})\\
=&\,\sum_{\mm \in \MM_{i|0:i-1}}w_\mm^{(i)'}\mu_{\mm,\vartheta_{i|0:i-1}}(y_i),
\end{aligned}
\end{equation} 
the last identity following from \eqref{marginals}. Hence the marginal likelihood of $y_{i}$ is
\begin{equation}\label{marginal likelihood}
\mu_{\nu_{i|0:i-1}}(y_{i})
:= \int_{\X}f_{x_{i}}(y_{i})\nu_{i|0:i-1}(x_{i}) = \sum_{\mm \in \MM_{i|0:i-1}}w_\mm^{(i)'}\mu_{\mm,\vartheta_{i|0:i-1}}(y_i)
\end{equation} 
with $w_\mm^{(i)'}$ as in \eqref{weights of filtering recursion} and $\mu_{\mm,\vartheta_{i|0:i-1}}$ as in \eqref{marginals}.
Dividing and multiplying by $\mu_{\oo,\theta_0}(y_0)= \int_{\X}f_{x_{0}}(y_{0})\pi(x_0)$ we can write
\begin{equation}
\int_{\X}P_{\Delta}(x_1\mid x_{0})f_{x_{0}}(y_0)\pi(\d x_0)
=\mu_{\oo, \theta_0}(y_{0})\int_{\X}P_{\Delta}(x_1\mid x_{0})\phi_{y_0}(\pi(\d x_0)).
\end{equation} 
The right hand side equals $\mu_{\oo,\theta_0}(y_{0})\nu_{1|0:1}(x_{1})$, so the equation for the likelihood becomes
\begin{equation}
\begin{aligned}
p(y_{0:T}) 
 = \mu_{\oo,\theta_0}(y_{0})
\int_{\X^{T}} \prod_{i=2}^{T} f_{x_{i}}(y_{i})P_{\Delta}(x_i\mid x_{i-1})f_{x_{1}}(y_{1})
\nu_{1|0:1}(x_{1})
\end{aligned}
\end{equation} 
Iterating with $\mu_{\nu_{i|0:i-1}}(y_i)= \int_{\X}f_{x_{i}}(y_{i})
\nu_{i|0:i-1}(x_{i})$ and $\phi_{y_{i}}(\nu_{i|0:i-1})$ in place of $\mu_{\oo, \theta_0}(y_0)$ and $\phi_{y_{0}}(\nu)$ leads to writing
\begin{equation}
p(y_{0:T}) =\mu_{\oo,\theta_0}(y_0)\prod_{i=1}^{T} \mu_{\nu_{i|0:i-1}}(y_{i})
\end{equation} 
which, through \eqref{marginal likelihood}, yields the result. \qed

\subsection*{Proof of \Cref{prop: prediction for smoothing}}

We prove \Cref{prop: prediction for smoothing} by induction. 
For $i=T-1$, we have
\begin{equation}\label{proof_smoothing}
\begin{aligned}
p(y_{T}\mid x_{T-1})
 = &\, \int_\X p(y_{T}\mid x_{T}) P_{\Delta}(x_{T}\mid x_{T-1})\d x_{T}\\
 = &\, \int_\X f_{x_{T}}(y_{T})h(x_t, \oo, \theta_0) P_{\Delta}(x_{T}\mid x_{T-1})\d x_{T} \\
 = &\, \mu_{\oo, \theta_0}(y_T) \int_\X h(x_t, t(y_{T},\oo), T(y_{T},\theta_0)) P_{\Delta}(x_{T}\mid x_{T-1})\d x_{T}\\
= &\, \mu_{\oo, \theta_0}(y_T)\mathbb{E}^{x_{T-1}}[h(X_{T}, \nn, \avt_{T:T})] \\
\end{aligned}
\end{equation} 
where in the penultimate equality we used the fact that using \Cref{A: conjugacy} and \eqref{marginals} allow to write
\begin{align}\label{eq:fh=muh}
\begin{aligned}
f_{x}(y)h(x, \mm, \theta) =&\,
\mu_{\mm, \theta}(y)\, h(x, t(y, \mm), T(y, \theta)),
\end{aligned}
\end{align} 
and in the last equality we let $\nn=t(y_{T},\oo)$ and $\avt_{T:T}$ is as in \eqref{backward lambda e vartheta}. 
Lemma \textcolor{blue}{S1.2} 
now implies 
\begin{equation}\label{gamtm1}
\begin{aligned}
p(y_{T}\mid x_{T-1})
 = &\, \sum_{\mm \in \aMM_{T-1:T}, \nn \in \aMM_{T:T}, \nn \ge \mm} 
\mu_{\oo, \theta_0}(y_T)p_{\nn,\mm}(\Delta;\avt_{T:T})
 h(x_{T-1}, \mm, \avt_{T-1:T}) 
\end{aligned}
\end{equation} 
with $\aMM_{T-1:T}=\B(\aMM_{T:T})=\B(\{\nn\})$. 
Hence the statement holds for $i=T-1$ with
$$ \overleftarrow w^{(T)}_{\mm}=\mu_{\oo, \theta_0}(y_T)p_{\nn,\mm}(\Delta;\avt_{T:T}), \quad \nn=t(y_{T},\oo),$$
since $\{\ii:\ \ii\in\aMM_{T:T}, \ii\ge \mm\}=\{\nn\}$ in \eqref{weights of prediction for smoothing}.
Assume now it holds for $p(y_{i+2:T}\mid x_{i+1})$, i.e.
\begin{equation}
 p(y_{i+2:T}\mid x_{i+1}) 
 = 
 \sum_{\mm \in \aMM_{i+1|i+2:T} } \overleftarrow w^{(i+2)}_{\mm}h(x_{i+1}, \mm, \avt_{i+1|i+2:T}).
\end{equation} 
Then
\begin{align}
p(&\,y_{i+1:T}\mid x_i)=\int_\X p(y_{i+1}\mid x_{i+1}) p(y_{i+2:T}\mid x_{i+1})P_{\Delta}(x_{i+1}\mid x_{i})\d x_{i+1} \\
=&\sum_{\mm \in \MM_{i+1|i+2:T} }\overleftarrow w_{\mm}^{(i+2)}\int_\X f_{x_{i+1}}(y_{i+1}) h(x_{i+1}, \mm, \avt_{i+1|i+2:T})P_{\Delta}(x_{i+1}\mid x_{i})\d x_{i+1}\\
=&\sum_{\mm \in \MM_{i+1|i+2:T} }\overleftarrow w_{\mm}^{(i+2)}
\mu_{\mm,\avt_{i+1|i+2:T}}(y_{i+1})\\
&\,\times
\int_\X h(x_{i+1}, t(y_{i+1},\mm), T(y_{i+1},\avt_{i+1|i+2:T}))P_{\Delta}(x_{i+1}\mid x_{i})\d x_{i+1}\\
=&\sum_{\mm \in \MM_{i+1|i+2:T} }\overleftarrow w_{\mm}^{(i+2)}
\mu_{\mm,\avt_{i+1|i+2:T}}(y_{i+1})
\E_{x_{i}}[h(X_{i+1}, t(y_{i+1},\mm), \avt_{i+1|i+1:T})]
\end{align}
where the third identity follows from \eqref{eq:fh=muh} and in the fourth we have used $T(y_{i+1},\avt_{i+1|i+2:T})=\avt_{i+1|i+1:T}$. Since $\MM_{i+1|i+1:T}=t(y_{i+1},\MM_{i+1|i+2:T})$, applying again 
Lemma \textcolor{blue}{S1.2} 
yields 
\begin{equation}
\begin{aligned}
\sum_{\mm \in \MM_{i+1|i+2:T} }&\,
\overleftarrow w_{\mm}^{(i+2)}
\mu_{\mm,\avt_{i+1|i+2:T}}(y_{i+1})
\E_{x_{i}}[h(X_{i+1}, t(y_{i+1},\mm), \avt_{i+1|i+1:T})] \\
=&\,\sum_{\mm \in \MM_{i+1|i+2:T} }
\overleftarrow w_{\mm}^{(i+2)}
\mu_{\mm,\avt_{i+1|i+2:T}}(y_{i+1}) \\
& \sum_{\nn \le t(y_{i+1},\mm)}p_{t(y_{i+1},\mm),\nn}(\Delta; \avt_{i+1|i+1:T})h(x_{i}, \nn, \avt_{i|i+1:T}) \\
 = &\, \sum_{\mm \in \MM_{i|i+1:T} } \\&
 \sum_{\nn\in\MM_{i+1|i+2:T}, \nn\ge \mm}
\overleftarrow w_{\nn}^{(i+2)}
\mu_{\nn,\avt_{i+1|i+2:T}}(y_{i+1})
p_{t(y_{i+1},\nn),\mm}(\Delta; \avt_{i+1|i+1:T})
h(x_{i}, \mm, \avt_{i|i+1:T}) 
\end{aligned}
\end{equation} 
where $\avt_{i|i+1:T}=\Theta_{\Delta}(\avt_{i+1|i+1:T})$. 
Finally, since $\MM_{i|i+1:T}=\B(\MM_{i+1|i+1:T})$, we obtain
\begin{equation}
\sum_{\mm \in \aMM_{i|i+1:T} } \overleftarrow w^{(i+1)}_{\mm}h(x_{i}, \mm, \avt_{i|i+1:T})
\end{equation} 
with 
\begin{equation}
\overleftarrow w^{(i+1)}_{\mm}
=\sum_{\nn\in\aMM_{i+1|i+2:T}, t(y_{i+1}, \nn) \ge \mm}
\overleftarrow w_{\nn}^{(i+2)}
\mu_{\nn,\avt_{i+1|i+2:T}}(y_{i+1})
p_{t(y_{i+1},\nn),\mm}(\Delta; \avt_{i+1|i+1:T}).
\end{equation} 
\qed

\subsection*{Proof of \Cref{thm: smoothing}}

From Propositions \ref{prop:rec_filtering} and \ref{prop: prediction for smoothing}, the numerator on the right hand side of \eqref{smoothing_expression} reads
\begin{equation}
 \sum_{\mm \in \aMM_{i|i+1:T} }\overleftarrow w^{(i+1)}_{\mm}h(x_{i}, \mm, \avt_{i|i+1:T})
 \sum_{\nn \in \MM_{i|i:T}}w_\nn^{(i)} h(x_{i}, \nn, \vartheta_{i|0:i})\pi(x_{i}).
\end{equation} 
\textbf{}
By \Cref{A: h-stability} the previous equals
\begin{equation}
 \sum_{\mm \in \aMM_{i|i+1:T} } 
 \sum_{\nn \in \MM_{i|i:T}}
 \overleftarrow w^{(i+1)}_{\mm}
 w_\nn^{(i)}C_{\mm, \nn, \avt_{i|i+1:T},\vartheta_{i|0:i}}
 h(x_{i}, d(\mm,\nn), e(\avt_{i|i+1:T},\vartheta_{i|0:i}))\pi(x_{i}).
\end{equation} 
Since
\begin{equation}
\int_{\X}h(x_{i}, d(\mm,\nn), e(\avt_{i|i+1:T},\vartheta_{i|0:i}))\pi(x_{i})=1
\end{equation} 
by \Cref{A: conjugacy} and the above sum is finite, the latter can be normalised, so \eqref{smoothing_expression} is the distribution
\begin{equation}
\begin{aligned}
p(x_i\mid y_{0:T})
=&\,
\sum_{\mm \in \aMM_{i|i+1:T},\nn \in \MM_{i|i:T}}
w_{\mm,\nn}^{(i)}g(x_{i}, d(\mm,\nn), e(\avt_{i|i+1:T},\vartheta_{i|0:i})),
\end{aligned}
\end{equation} 
with 
\begin{equation}
w_{\mm,\nn}^{(i)}=\frac{\overleftarrow w^{(i+1)}_{\mm} w_\nn^{(i)}C_{\mm, \nn, \avt_{i|i+1:T},\vartheta_{i|0:i}}}
{\sum_{\ii \in \aMM_{i|i+1:T},\jj \in \MM_{i|i:T}}
\overleftarrow w^{(i+1)}_{\ii} w_{\jj}^{(i)}C_{\ii, \jj, \avt_{i|i+1:T},\vartheta_{i|0:i}}}.
\end{equation} 
\qed

\subsection*{Proof of \Cref{thm: WF_joint_smoothing}}

The forward kernel \eqref{forward kernel} 
involves the cost-to-go functions, available from \Cref{prop: prediction for smoothing}, the Multinomial emission density, and the transition density \eqref{eq: WF transition dens} of the \gls{WF} process.
Then we can write
\begin{align}
   p(&\xx_{i}|\xx_{i-1},\yy_{i:T}) =  p(\yy_{i:T}|\xx_{i-1})^{-1}p(\xx_{i}|\xx_{i-1})p(\yy_{i}|\xx_{i})p(\yy_{i+1:T}|\xx_{i})  \\ 
=&\,p(\yy_{i:T}|\xx_{i-1})^{-1}\sum_{m=0}^\infty q_m \sum_{\ll \in \N^K; l_1+\ldots+l_K = m} \text{Multinom}(\ll|m,\xx_{i-1})\Dir(\xx_i\mid\ll + \xx_{i-1}+\yy_i) \\ &\frac{\Gamma(|\yy_i|)\mathcal{B}(\ll + 
   \xx_{i-1}+\yy_i)}{\prod_j \Gamma(y_{ij})\mathcal{B}(\ll + 
   \xx_{i-1})} \sum_{\kk \in \aMM_{i|i+1:T}}
 \overleftarrow w^{(i+1)}_{\kk}h(\xx_{i}, \kk) \\
=&\,p(\yy_{i:T}|\xx_{i-1})^{-1}\sum_{m=0}^\infty q_m \sum_{\ll \in \N^K; l_1+\ldots+l_K = m} \sum_{\kk \in \aMM_{i|i+1:T}}\text{Multinom}(\ll|m,\xx_{i-1})\\ &\, \times \frac{\Gamma(|\yy_i|)\mathcal{B}(\ll + 
   \xx_{i-1}+\yy_i)}{\prod_j \Gamma(y_{ij})\mathcal{B}(\ll + 
   \xx_{i-1})} 
 \overleftarrow w^{(i+1)}_{\kk}h(\xx_{i}, \kk)\Dir(\xx_i\mid\ll + \xx_{i-1}+\yy_i)\\
=&\,p(\yy_{i:T}|\xx_{i-1})^{-1}\sum_{m=0}^\infty q_m \sum_{\ll \in \N^K; l_1+\ldots+l_K = m} \sum_{\kk \in \aMM_{i|i+1:T}}\text{Multinom}(\ll|m,\xx_{i-1})\\ &\, \times  \frac{\Gamma(|\yy_{i}|)}{\prod_j \Gamma(y_{ij})}  \frac{\Gamma(|\aa|+|\kk|)\prod_j \Gamma(\alpha_j)}{\Gamma(|\aa|)\prod_j \Gamma(\alpha_j+k_j)}
   \frac{\mathcal{B}(\ll + 
   \xx_{i-1}+\yy_i+\kk)}{\mathcal{B}(\ll + 
   \xx_{i-1})}
 \overleftarrow w^{(i+1)}_{\kk} \\
 &\, \times\Dir(\xx_i\mid\ll + \xx_{i-1}+\yy_i+\kk).
 \end{align}
Denote now 
\begin{align}
w_{\ll, \kk} &= \mathcal{B}(\ll + 
   \xx_{i-1}+\yy_i+\kk)\frac{\Gamma(|\aa|+|\kk|)}{\prod_j \Gamma(\alpha_j+k_j)}  \overleftarrow w^{(i+1)}_{\kk} \\
   r_{\ll} & = \sum_{\kk \in \aMM_{i|i+1:T}} w_{\ll, \kk}\\
   \tilde w_{\ll, \kk} & = r_{\ll}^{-1}w_{\ll, \kk}\\
   r_m & = \sum_{\ll \in \N^K; l_1+\ldots+l_K = m} r_{\ll}\frac{\text{Multinom}(\ll|m,\xx_{i-1})}{\mathcal{B}(\ll + 
   \xx_{i-1})}  \\
   \tilde  w_{\ll}^m &= r_m^{-1}r_{\ll}\frac{\text{Multinom}(\ll|m,\xx_{i-1})}{\mathcal{B}(\ll + 
   \xx_{i-1})}\\
   \tilde{q}_m & = p(\yy_{i:T}|\xx_{i-1})^{-1} \frac{\Gamma(|\yy_{i}|)\prod_j \Gamma(\alpha_{j})}{\prod_j \Gamma(y_{ij})\Gamma(|\alpha|)} r_m q_m = \tilde r_m q_m.
\end{align} 
We can then write  $   p(\xx_{i}|\xx_{i-1},\yy_{i:T})$ as the mixture
\begin{align} \label{eq: WF forward mixture form}
 p(x_{i}|x_{i-1},y_{i:T}) = \sum_{m=0}^\infty \tilde{q}_m \sum_{\ll \in \N^K; l_1+\ldots+l_K = m} \tilde  w_{\ll}^m \sum_{\kk \in \aMM_{i|i+1:T}}\tilde w_{\ll, \kk}\Dir(\xx_i|\ll + \xx_{i-1}+\yy_i+\kk),
\end{align}
where, by construction,
\begin{align}
&\forall m \in \N,  \forall \ll \in \N^K \text{ such that } l_1+\ldots+l_K = m:\  \sum_{\kk \in \aMM_{i|i+1:T}} \tilde w_{\ll, \kk} =  1\\
&\forall m \in \N:\  \sum_{\ll \in \N^K; l_1+\ldots+l_K = m} \tilde  w_{\ll}^m =  1.
\end{align} 
As the right hand side of \eqref{eq: WF forward mixture form} is a sum of integrable positive functions, integrating both sides with respect to $\xx_i$ implies
\begin{equation}
 \sum_{m = 0}^{\infty} \tilde q_m = 1 \iff \sum_{m \ge 0} r_mq_m = p(y_{i:T}|x_{i-1}) \frac{\prod_j \Gamma(y_{ij})\Gamma(|\alpha|)}{\Gamma(|\yy_{i}|)\prod_j \Gamma(\alpha_{j})}
\end{equation}
implying the forward kernel is an infinite mixture of finite mixtures of Dirichlet distributions with respect to $m$, $\ll$ and $\kk$.

\qed

\subsection*{Sampling from WF the forward kernel}

The two inner sums in \eqref{eq: WF_forward_kernel} are finite mixtures, so the challenge to sample from the \gls{WF} forward kernel lies in how to sample from the infinite mixture with weights $\tilde{q}_m$. 
To this end, let $b_{i}^{(t, \theta)}(m)=a_{i m}^{\theta} e^{-i(i+\theta-1) t / 2}$ where $a_{i m}^{\theta}=\frac{(\theta+2 i-1)(\theta+m)_{(i-1)}}{m !(i-m) !}$. Then the $q_m$ in \eqref{eq: WF transition dens} can be written
\begin{equation}
 q_{m}=\sum_{i=0}^{\infty}(-1)^{i} b_{i}(m).
\end{equation}
Then the $\tilde q_m$ in \eqref{eq: WF_forward_kernel} can be written as $\tilde r_m q_m$ where, for any given $m$, $\tilde r_m$ is a strictly positive real number (cf.~Proof of \Cref{thm: WF_joint_smoothing}). 
Similarly, we define $\tilde b_{i}^{(t, \theta)}(m) = \tilde r_m b_{i}^{(t, \theta)}(m)$. 

Note now that all the properties allowing to sample from the distribution defined by the weights $q_m$ in Proposition 1 of \cite{jenkins2017exact} carry over to the distribution defined by the weights $\tilde q_m$. 
In fact, the following proposition holds:
\begin{proposition}\label{prop: jenkins_adapted}
Let \begin{equation}
 \tilde C_{m}^{(t, \theta)}:=\inf \left\{j \geq 0: \tilde b_{j+m+1}^{(t, \theta)}(m)< \tilde b_{j+m}^{(t, \theta)}(m)\right\}.
\end{equation}
Then
\begin{enumerate}
 \item $ \tilde C_{m}^{(t, \theta)} < \infty$, for all $m$.
 \item $\tilde b_{i}^{(t, \theta)}(m) \downarrow 0 \text { as } i \rightarrow \infty \text { for all } i \geq m+\tilde C_{m}^{(t, \theta)}$.
 \item $\tilde C_{m}^{(t, \theta)}=0 \text { for all } m>D_{0}^{(t, \theta)}, \text { where for } \varepsilon \in[0,1)$,
 \begin{equation}
 D^{(t, \theta)}:=\inf \left\{i \geq\left(\frac{1}{t}-\frac{\theta+1}{2}\right) \vee 0:(\theta+2 i+1) e^{-\frac{(2 i+\theta) t}{2}}<1\right\}.
\end{equation}
\end{enumerate}

\end{proposition}
\begin{proof}
Let $b_{i}^{(t, \theta)}(m)=a_{i m}^{\theta} e^{-i(i+\theta-1) t / 2}$ where $a_{i m}^{\theta}=\frac{(\theta+2 i-1)(\theta+m)_{(i-1)}}{m !(i-m) !}$, and let
\begin{equation}
 C_{m}^{(t, \theta)}:=\inf \left\{i \geq 0: b_{i+m+1}^{(t, \theta)}(m)<b_{i+m}^{(t, \theta)}(m)\right\}.
\end{equation}
Then Proposition 1 in \cite{jenkins2017exact} implies
\begin{enumerate}
 \item $C_{m}^{(t, \theta)}<\infty,$ for all $m$
 \item $b_{i}^{(t, \theta)}(m) \downarrow 0$ as $i \rightarrow \infty$ for all $i \geq m+C_{m}^{(t, \theta)} ;$ and
 \item $C_{m}^{(t, \theta)}=0$ for all $m>D_{0}^{(t, \theta)},$ where for $\varepsilon \in[0,1)$,
 \begin{equation}
 D_{0}^{(t, \theta)}:=\inf \left\{i \geq\left(\frac{1}{t}-\frac{\theta+1}{2}\right) \vee 0:(\theta+2 i+1) e^{-\frac{(2 i+\theta) t}{2}}<1\right\}.
\end{equation}
\end{enumerate}
Note now that  $\tilde C_{m}^{(t, \theta)} = C_{m}^{(t, \theta)}$, since
\begin{align}
 \tilde C_{m}^{(t, \theta)}:=&\inf \left\{j \geq 0: \tilde b_{j+m+1}^{(t, \theta)}(m)< \tilde b_{j+m}^{(t, \theta)}(m)\right\} \\
 = & \inf \left\{j \geq 0: \tilde r_m b_{j+m+1}^{(t, \theta)}(m)<  \tilde r_m b_{j+m}^{(t, \theta)}(m)\right\} \\
 = & \inf \left\{j \geq 0:  b_{j+m+1}^{(t, \theta)}(m)<   b_{j+m}^{(t, \theta)}(m)\right\}
 = C_{m}^{(t, \theta)}
\end{align}
Next, observe that
\begin{enumerate}
 \item $C_{m}^{(t, \theta)}<\infty,$ for all $m $ implies $\tilde C_{m}^{(t, \theta)}<\infty,$ for all $m$.
 \item $b_{i}^{(t, \theta)}(m) \downarrow 0$ as $i \rightarrow \infty$ for all $i \geq m+C_{m}^{(t, \theta)}$ implies that $ \tilde b_{i}^{(t, \theta)}(m) \downarrow 0$ as $i \rightarrow \infty$ for all $i \geq m+\tilde C_{m}^{(t, \theta)} $, as $\tilde b_{i}^{(t, \theta)}(m) = \tilde r_m b_{i}^{(t, \theta)}(m)$ with $\tilde r_m$ a positive constant.
 \item  $C_{m}^{(t, \theta)}=0$ for all $m>D_{0}^{(t, \theta)},$ where for $\varepsilon \in[0,1) $ implies $ \tilde C_{m}^{(t, \theta)}=0$ for all $m>D_{0}^{(t, \theta)},$ where for $\varepsilon \in[0,1)$.
\end{enumerate}
This completes the proof.
\end{proof}

The interpretation of Proposition \ref{prop: jenkins_adapted} is that once the coefficients $\tilde b_{i}^{(t, \theta)}(m)$ start to decay, they keep decaying indefinitely so it becomes possible to bound them and use the alternated series trick from \cite{Devroye1986a}.

To this end, let $\ii = (i_0, i_1, \ldots, i_M)$ and define
\begin{equation}
 \tilde S_{\ii}^{-}(M):=\sum_{m=0}^{M} \sum_{j=0}^{2 i_{m}+1}(-1)^{j} \tilde b_{m+j}^{(t, \theta)}(m), \quad \tilde S_{\ii}^{+}(M):=\sum_{m=0}^{M} \sum_{j=0}^{2 i_{m}}(-1)^{j} \tilde b_{m+j}^{(t, \theta)}(m).
\end{equation}
$\tilde S_{\ii}^{-}$ and $\tilde S_{\ii}^{+}$ form the two convergent series bracketing the target distribution, which in turn allows to use the alternated series method.
Algorithm \ref{algo_WF_sampling} details how to sample from the \gls{WF} forward kernel \eqref{eq: WF_forward_kernel}.

\begin{algorithm}[t]
\SetAlgoLined

\KwIn{$\xx_{i-1}$, $\yy_{i:T}$ and the cost-to-go weights $\overleftarrow{w}_{\kk}^{(i+1)}$}

\KwResult{A sample $\xx_{i}$ from the forward kernel}

\SetKwBlock{Begin}{Initialise}{}

\Begin{

Set $m \longleftarrow 0, k_{0} \longleftarrow 0, \ii \longleftarrow\left(i_{0}\right)$\\

Simulate $U \sim \text { Uniform[0, } 1]$\\

}

\While{TRUE}{

Set $i_{m} \longleftarrow\left\lceil \tilde C_{m}^{(t, \theta)} / 2\right\rceil$ with $\tilde C_{m}^{(t, \theta)}$ as in Proposition \ref{prop: jenkins_adapted}.

\While{$\tilde S_{\mathbf{i}}^{-}(m)<U<S_{\mathbf{i}}^{+}(m)$}{$\text { Set } \mathbf{i} \leftarrow \mathbf{i}+(1,1, \ldots, 1)$}%

\uIf{$S_{\mathbf{i}}^{-}(m)>U$}{

\textbf{return} m

}
\uElseIf{$S_{\mathbf{i}}^{+}(m)<U$}{Set $\mathbf{i} \longleftarrow\left(i_{0}, i_{1}, \dots, i_{m}, 0\right)$ \\ Set $m \leftarrow m+1$}



\textbf{end} \\
Simulate $\ll \sim \text{Categorical}\left(\{\tilde w^m_{\ll};  \ll \in \mathbb{N}^{K} ; l_{1}+\ldots+l_{K}=m\}\right)$ \\
Simulate $\kk \sim \text{Categorical}\left(\{\tilde w_{\ll, \kk}; \kk \in \aMM_{i|i+1:T} \}\right)$ \\
Simulate $\xx_i \sim \Dir(\ll + \xx_{i-1}+\yy_i+\kk)$
}
\begin{quote}
\caption{\small Simulating from the Wright--Fisher forward kernel\label{algo_WF_sampling}}
\end{quote}
\end{algorithm}

\bibliographystyle{apalike}
\bibliography{input_files/library}

\end{document}


\title{Supplementary Information for ``Exact inference for a class of non-linear hidden Markov models on general state spaces''}

\author{Guillaume Kon Kam King \\ Universit\'e Paris-Saclay, INRAE, MaIAGE\\ 78350, Jouy-en-Josas, France \\ guillaume.kon-kam-king@inrae.fr\\
\vspace{0.5em}\\
Omiros Papaspiliopoulos\\ ICREA and Department of Economics and Business\\ Universitat Pompeu Fabra \\ Ram\'on Trias Fargas 25-27, 08005, Barcelona, Spain \\ omiros.papaspiliopoulos@upf.edu\\
\vspace{0.5em}\\
Matteo Ruggiero\\ University of Torino and Collegio Carlo Alberto \\
Corso Unione Sovietica 218/bis, 10134, Torino, Italy \\ matteo.ruggiero@unito.it}

\maketitle
\newpage
\tableofcontents 

\setcounter{equation}{0}
\setcounter{figure}{0}
\setcounter{table}{0}
\setcounter{page}{1}
\setcounter{section}{0}
\makeatletter
\renewcommand{\theequation}{S\arabic{equation}}
\renewcommand{\thefigure}{S\arabic{figure}}
\renewcommand{\thesection}{S\arabic{section}}
\renewcommand{\thealgocf}{S\arabic{algocf}}
\renewcommand{\bibnumfmt}[1]{[S#1]}
\renewcommand{\citenumfont}[1]{S#1}
\newtheorem{applemma}{Lemma}
\renewcommand{\theapplemma}{\thesection.\arabic{applemma}}
\setcounter{applemma}{0}

\newpage

\section{Proofs of some lemmas}\label{sec:proofs}

\begin{applemma}\label{lemma transition probabilities}

The transition probabilities
(\textcolor{blue}{$7$}) 
are
\begin{equation}
p_{\mm, \mm - \ii}(t, \theta) = \gamma_{\norm{\mm}, \norm{\ii} }C_{\norm{\mm}, \norm\mm - \norm\ii }(t) \mathrm{MVH}(\ii; \mm, \norm{\ii} ) \label{def_pmmi}
\end{equation}
where $\gamma_{\norm{\mm}, \norm{\ii} } = \prod_{h=0}^{\norm{\ii} -1}\lambda_{\norm{\mm} -h}$,
\begin{equation}\label{gamma mi cmmi}
C_{\norm{\mm}, \norm{\mm} - \norm{\ii} }(t) = (-1)^{\norm{\ii} }\sum_{k=0}^{\norm{\ii} }\frac{e^{-\lambda_{\norm{\mm} -k}\int_0^t \rho(\Theta_s)\d s}}{\prod_{0\le h \le \norm{\ii}, h\ne k}(\lambda_{\norm{\mm} -k}-\lambda_{\norm{\mm} -h})} 
\end{equation}
and $\mathrm{MVH}(\ii;\mm,\norm{\ii} )$ is the multivariate hypergeometric pmf evaluated at $\ii$, with parameters $\mm$ and $\norm{\ii}$.
\end{applemma}
\begin{proof}
The proof can be found in \cite{Papaspiliopoulos2014a}, Proposition 2.1.
\end{proof}

\begin{applemma}\label{h_expect}
Let 
Assumptions \textcolor{blue}{1-2} 
hold. Then 
\begin{equation}
 \mathbb{E}^{x}\left[h(X_t, \mm, \theta)\right] = \sum_{\nn \le \mm}p_{\mm, \nn}(t; \theta)h(x, \nn, \Theta_t) 
\end{equation}
with $\Theta_t$ being the unique solution to 
(\textcolor{blue}{$6$}) 
with $\Theta_0=\theta$.
\end{applemma}
\begin{proof}
The statement follows from an application of 
(\textcolor{blue}{$8$}) 
with $\Theta_{0}=\theta$ and by noting that 
\begin{equation}
\begin{aligned}
\mathbb{E}^{\mm, \theta}\left[h(x, M_t, \Theta_t)\right] 
 = &\, \sum_{\nn\in \Z_{+}^{K}}p_{\mm,\nn}(t,\theta)h(x, \nn, \Theta_t)
 = \sum_{\nn \le \mm}p_{\mm, \nn}(t; \theta)h(x, \nn, \Theta_t) 
\end{aligned}
\end{equation} 
where $p_{\mm, \nn}(t; \theta)$ are as in 
(\textcolor{blue}{$7$})
.
\end{proof}

\section{Exact $L_2$ distances}\label{secA: Exact_L2_distances}

\subsection{$L_2$ distance between mixtures of gamma distributions}

\begin{lemma}\label{L2gamma}
 Consider two gamma mixtures $g = \sum_{i=1}^I g_i f_i$ and $h = \sum_{j=1}^J h_j f_j$ where $\forall i \in \N; f_i := \Ga(\alpha_i, \beta_i)$. 
 Let us further assume that $\forall i \in \N; \alpha_i>0.5$.
 Then:
 
 \begin{align}
 \int_{\R^+}(g-h)^2  = & \sum_{i, j=1}^I g_ig_j \frac{\beta_1^{g~\alpha_1^g}\beta_2^{g~\alpha_2^g}}{\Gamma(\alpha_1^g)\Gamma(\alpha_2^g)}\frac{\Gamma(\alpha_1^g+\alpha_2^g-1)}{(\beta_1^g+\beta_2^g)^{(\alpha_1^g+\alpha_2^g-1)}}  \\
 & \sum_{i, j=1}^J h_ih_j  \frac{\beta_1^{h~\alpha_1^h}\beta_2^{h~\alpha_2^h}}{\Gamma(\alpha_1^h)\Gamma(\alpha_2^h)}\frac{\Gamma(\alpha_1^h+\alpha_2^h-1)}{(\beta_1^h+\beta_2^h)^{(\alpha_1^h+\alpha_2^h-1)}} \\ 
 & - 2 \sum_{i=1}^I \sum_{j=1}^J g_i h_j \frac{\beta_1^{g~\alpha_1}\beta_2^{h~\alpha_2}}{\Gamma(\alpha_1^g)\Gamma(\alpha_2^h)}\frac{\Gamma(\alpha_1^g+\alpha_2^h-1)}{(\beta_1^g+\beta_2^h)^{(\alpha_1^g+\alpha_2^h-1)}}.
\end{align}
\end{lemma}

\begin{proof}
 
 We denote $\forall i \in \N; f_i := \Ga(\alpha_i, \beta_i)$

\begin{align}
 \int_{\R^+} f_1f_2 & =  \int_{\R^+} \frac{\beta_1^{\alpha_1}\beta_2^{\alpha_2}}{\Gamma(\alpha_1)\Gamma(\alpha_2)}x^{(\alpha_1+\alpha_2 - 1)-1}e^{-(\beta_1+\beta_2)x}\\
  & =   \frac{\beta_1^{\alpha_1}\beta_2^{\alpha_2}}{\Gamma(\alpha_1)\Gamma(\alpha_2)}\frac{\Gamma(\alpha_1+\alpha_2-1)}{(\beta_1+\beta_2)^{(\alpha_1+\alpha_2-1)}}\int_{\R^+} \Ga(\alpha_1 + \alpha_2-1, \beta_1 + \beta_1) 
\end{align}

Now provided that $\alpha_1+\alpha_2 > 1$, which is implied by $\forall i \in \N; \alpha_i>0.5$:

\begin{equation}
   \int_{\R^+} f_1f_2 =   \frac{\beta_1^{\alpha_1}\beta_2^{\alpha_2}}{\Gamma(\alpha_1)\Gamma(\alpha_2)}\frac{\Gamma(\alpha_1+\alpha_2-1)}{(\beta_1+\beta_2)^{(\alpha_1+\alpha_2-1)}}
\end{equation}

We now consider two gamma mixtures $g = \sum_{i=1}^I g_i f_i$ and $h = \sum_{j=1}^J h_j f_j$.

\begin{align}
 (g-h)^2 & =  \left(\sum_{i=1}^I g_i f_i - \sum_{j=1}^J h_j f_j \right)^2 \\
 & =  \left(\sum_{i=1}^I g_i f_i \right)^2 + \left(\sum_{j=1}^J h_j f_j \right)^2 - 2 \sum_{i=1}^I g_i f_i  \sum_{j=1}^J h_j f_j \\
  & =  \sum_{i, j=1}^I g_ig_j f_if_j  + \sum_{i, j=1}^J h_ih_j f_if_j - 2 \sum_{i=1}^I \sum_{j=1}^J g_i h_j f_if_j 
\end{align}

\begin{equation} \label{product_expansion}
 \int_{\R^+}(g-h)^2  =  \sum_{i, j=1}^I g_ig_j \int_{\R^+} f_if_j  + \sum_{i, j=1}^J h_ih_j \int_{\R^+} f_if_j - 2 \sum_{i=1}^I \sum_{j=1}^J g_i h_j \int_{\R^+} f_if_j  
\end{equation}

Let us now write $g = \sum_{i=1}^I g_i f_i^g = \Ga(\alpha_i^g, \beta_i^g)$ and $h = \sum_{i=1}^I h_i f_i^h = \Ga(\alpha_i^h, \beta_i^h)$

\begin{align}
 \int_{\R^+}(g-h)^2  = & \sum_{i, j=1}^I g_ig_j \frac{\beta_1^{g~\alpha_1^g}\beta_2^{g~\alpha_2^g}}{\Gamma(\alpha_1^g)\Gamma(\alpha_2^g)}\frac{\Gamma(\alpha_1^g+\alpha_2^g-1)}{(\beta_1^g+\beta_2^g)^{(\alpha_1^g+\alpha_2^g-1)}}  + \sum_{i, j=1}^J h_ih_j  \frac{\beta_1^{h~\alpha_1^h}\beta_2^{h~\alpha_2^h}}{\Gamma(\alpha_1^h)\Gamma(\alpha_2^h)}\frac{\Gamma(\alpha_1^h+\alpha_2^h-1)}{(\beta_1^h+\beta_2^h)^{(\alpha_1^h+\alpha_2^h-1)}} \\ 
 & - 2 \sum_{i=1}^I \sum_{j=1}^J g_i h_j \frac{\beta_1^{g~\alpha_1}\beta_2^{h~\alpha_2}}{\Gamma(\alpha_1^g)\Gamma(\alpha_2^h)}\frac{\Gamma(\alpha_1^g+\alpha_2^h-1)}{(\beta_1^g+\beta_2^h)^{(\alpha_1^g+\alpha_2^h-1)}}
\end{align}

\end{proof}

\subsection{$L_2$ distance between mixtures of Dirichlet distributions}

\begin{lemma}
 Consider two Dirichlet mixtures $g = \sum_{i=1}^I g_i f_i$ and $h = \sum_{j=1}^J h_j f_j$ where $\forall i \in \N; f_i := \Dir(\aa_i) = \frac{1}{B(\aa_i)}\prod_{j=1}^Kx_j^{\alpha_{i,j}-1}$ where $B(\aa_i) = \frac{\prod_{j=1}^K\Gamma(\alpha_{i,j})}{\Gamma(\sum_{j=1}^K\alpha_{i,j})}$. 
 Let us denote $\nabla_K$ the K-dimensional simplex.
 Let us further assume that $\forall i \in \N; \alpha_i>0.5$.
 Then:
 
\begin{align}
  \int_{\nabla_K} (g-h)^2  = & \sum_{i, j=1}^I g_ig_j \frac{B(\aa_i^g+\aa_j^g-1)}{B(\aa_i^g)B(\aa_j^g)}  + \sum_{i, j=1}^I h_ih_j \frac{B(\aa_i^h+\aa_j^h-1)}{B(\aa_i^h)B(\aa_j^h)} \\ 
 & - 2 \sum_{i=1}^I \sum_{j=1}^J g_i h_j \frac{B(\aa_i^g+\aa_j^h-1)}{B(\aa_i^g)B(\aa_j^h)}
\end{align}
\end{lemma}

\begin{proof}

We denote $\forall i \in \N; f_i := \Dir(\aa_i) = \frac{1}{B(\aa_i)}\prod_{j=1}^Kx_j^{\alpha_{i,j}-1}$ where $B(\aa_i) = \frac{\prod_{j=1}^K\Gamma(\alpha_{i,j})}{\Gamma(\sum_{j=1}^K\alpha_{i,j})}$ and $\nabla_K$ the K-dimensional simplex.

\begin{align}
 \int_{\nabla_K} f_1f_2 & =  \int_{\nabla_K} \frac{1}{B(\aa_1)B(\aa_2)}\prod_{j=1}^Kx_j^{\alpha_{1,j} + \alpha_{2,j}-2}\\
  & =   \frac{B(\aa_1+\aa_2-1)}{B(\aa_1)B(\aa_2)}\int_{\nabla_K} \Dir(\aa_1 + \aa_2-1)
\end{align}

Now provided that $\forall j \in \{1,\ldots, K\}, \alpha_{1,j} + \alpha_{2,j} > 1$:

\begin{equation}
   \int_{\nabla_K} f_1f_2 =   \frac{B(\aa_1+\aa_2-1)}{B(\aa_1)B(\aa_2)}
\end{equation}

We now consider two Dirichlet mixtures $g = \sum_{i=1}^I g_i f_i$ and $h = \sum_{j=1}^J h_j f_j$.

Following the product expansion in the previous subsection, we can also obtain the exact expression:

\begin{align}
  \int_{\nabla_K} (g-h)^2  = & \sum_{i, j=1}^I g_ig_j \frac{B(\aa_i^g+\aa_j^g-1)}{B(\aa_i^g)B(\aa_j^g)}  + \sum_{i, j=1}^I h_ih_j \frac{B(\aa_i^h+\aa_j^h-1)}{B(\aa_i^h)B(\aa_j^h)} \\ 
 & - 2 \sum_{i=1}^I \sum_{j=1}^J g_i h_j \frac{B(\aa_i^g+\aa_j^h-1)}{B(\aa_i^g)B(\aa_j^h)}
\end{align} 
\end{proof}

\section{Additional algorithm}

\begin{algorithm}[t]
\SetAlgoLined

\textbf{Pruning setting}: ON (approximate filtering) / OFF (exact filtering)

\KwIn{$Y_{0:n}$, $t_{0:n}$ and $\nu = h(x, \oo, \theta_0) \in \mathcal{F}$ for some $\theta_0 \in \Th$}

\KwResult{$\vartheta_{i|0:i}$, $\MM_{i|0:i}$, $W_{0:n}$ with $W_i = \{w_\mm^i, \mm \in \MM_{i|0:i}\}$ and $\vartheta_{i|0:i-1}$, $\MM_{i|0:i-1}$, $W'_{1:n}$ with $W'_i = \{w_\mm^{i'}, \mm \in \MM_{i|0:i-1}\}$ and the likelihood $p(y_{0:T})$.}

\SetKwBlock{Begin}{Initialise}{}

\Begin{

Set $\vartheta_{0|0} = T(Y_0, \theta_0)$ with $T$ as in 
(\textcolor{blue}{$12$})
\\

Set $\MM_{0|0} = \{t(Y_0, \oo)\} = \{\mm^*\}$ and $W_0 = \{1\}$ with $t$ as in Assumption \textcolor{blue}{$3$}
\\

Compute $\vartheta_{1|0}$ from $\vartheta_{0|0}$ as in (\textcolor{blue}{$12$})
\\

Set $\MM_{1|0} = \B(\MM_{0|0})$ and $W'_1 = \{p_{\mm^*, \nn}(\Delta, \vartheta_{0|0}), \nn \in \MM_{1|0}\}$ with $\B$ as in 
(\textcolor{blue}{$5$}) 
and $p_{\mm, \nn}$ as in
\eqref{def_pmmi}

}

\For{$i$ from $1$ to $n$}{

\SetKwBlock{Begin}{Update}{}

\Begin{

Set $\vartheta_{i|0:i} = T(Y_i, \vartheta_{i|0:i-1})$\\

Set $W_{i} = \{\frac{w_\mm^{i'} \mu_{\mm, \vartheta_{i|0:i-1}}(Y_i)}{\sum_{\nn \in \MM_{i|0:i-1}}w_\nn^{i'} \mu_{\nn, \vartheta_{i}}(Y_i)}, \mm \in \MM_{i|0:i-1}\}$ with $\mu_{\mm, \theta}$ defined as in 
(\textcolor{blue}{$11$})

Set $\MM_{i|0:i} = \{t(Y_i, \mm), \mm \in \MM_{i|0:i-1}\}$ and update the labels in $W_i$\\

}

\uIf{pruning \emph{ON}}{

Prune $\MM_{i|0:i}$ and remove the corresponding weights in $W_i$\\

Normalise the weights in $W_i$

}

\SetKwBlock{Begin}{Predict}{}

\Begin{

Compute $\vartheta_{i+1|0:i}$ from $\vartheta_{i|0:i}$\\

Set $\MM_{i+1|0:i} = \B(\MM_{i|0:i})$ and $W'_{i+1} = \left\{\displaystyle{\sum_{\mm \in \MM_{i|0:i}, \mm \ge \nn}}w_{\mm}^ip_{\mm, \nn}(\Delta,\vartheta_{i|0:i}), \nn \in \MM_{i+1|0:i}\right\}$
}

}
\SetKwBlock{Begin}{Compute likelihood}{end}

\Begin{
Compute the likelihood using $\vartheta_{i|0:i-1}$, $\MM_{i|0:i-1}$ and $W'_{1:n}$ as in 
Equation (\textcolor{blue}{$16$})

}
\begin{quote}
\caption{\small Filtering and likelihood\label{algo_filtering_and_lik}}
\end{quote}
\end{algorithm}

\section{Autocorrelation function plots for all parameters}

\includegraphics[width = \textwidth]{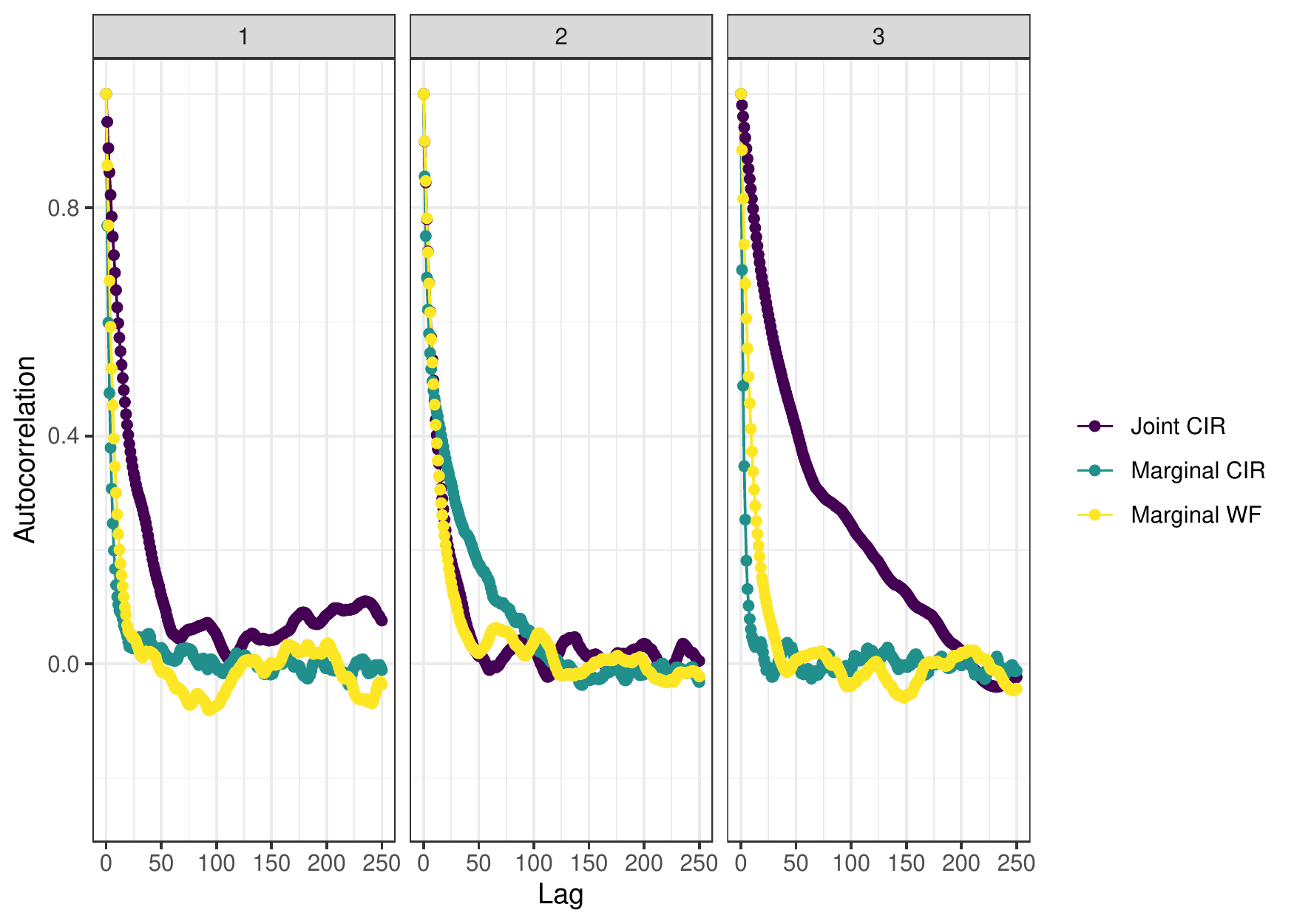}
\captionof{figure}{Autocorrelation function of the \gls{MCMC} chains for $(\alpha, \beta, \sigma)$ (\gls{CIR}) and for $\aa$ (\gls{WF}). The darkest line corresponds to the joint \gls{CIR} inference, the intermediate line corresponds to the marginal \gls{CIR} inference and the lightest line to the marginal \gls{WF} inference.\label{fig:ACF_plots_all_pars}
}

\bibliographystyle{apalike}
\bibliography{input_files/library}